\documentclass[a4paper, 11pt, oneside]{article}

\usepackage{fullpage}     
\usepackage{complexity}   
\usepackage{graphicx}     
\usepackage{amsmath,amsthm,amssymb}  
\usepackage{fancybox}     
\usepackage{enumerate}    
\usepackage{calc}         
\usepackage{xspace}       
\usepackage[%
  pdftitle={The star and biclique coloring and choosability problems},%
  pdfauthor={Marina Groshaus, Francisco J.\ Soulignac, Pablo Terlisky},%
  pdfcreator={},%
  pdfsubject={The star and biclique coloring and choosability problems},%
  pdfkeywords={star-coloring, biclique-coloring, star-choosability, biclique-choosability, star-chromatic number, biclique-chromatic number},%
  colorlinks=true,%
  linkcolor=blue,%
  citecolor=blue]{hyperref}  

{\theoremstyle{definition} \newtheorem{defn}{Definition}}
\newtheorem{theorem}{Theorem}
\newtheorem{observation}[theorem]{Observation}
\newtheorem{lemma}[theorem]{Lemma}
\newtheorem{corollary}[theorem]{Corollary}
\newtheorem{openproblem}{Problem}

\let\oldNP\NP\renewcommand{\NP}{\oldNP\xspace}

\newcommand{\bccol}[1]{\textsc{biclique $#1$-coloring}}
\newcommand{\stcol}[1]{\textsc{star $#1$-coloring}}
\newcommand{\stchose}[1]{\textsc{star $#1$-choosability}}
\newcommand{\bcchose}[1]{\textsc{biclique $#1$-choosability}}
\newcommand{\ptwop}{\ensuremath{\Pi^p_2}\xspace}
\newcommand{\stp}{\ensuremath{\Sigma^p_2}\xspace}
\newcommand{\ptp}{\ensuremath{\Pi^p_3}\xspace}
\newcommand{\qsat}[1]{\ensuremath{\textsc{qsat}_{#1}}}
\newcommand{\range}[3]{\ensuremath{#1 \in \{#2,\ldots,#3\}}}
\newcommand{\MC}{\ensuremath{\mathcal{C}}}
\newcommand{\MB}{\ensuremath{\mathcal{B}}}
\newcommand{\MS}{\ensuremath{\mathcal{S}}}

\newcommand{\naesat}{\textsc{nae-sat}\xspace}
\newcommand{\naesatt}{\textsc{nae$\forall\exists$sat}\xspace}

\let\VP=\mathbf 
\def\VEC#1{\vec{\VP #1}}

\newenvironment{Problem}{%
  \begin{Sbox}%
    \begin{minipage}{\textwidth-2\parindent}%
      \parskip=.1\baselineskip%
      \vspace{.25\baselineskip}%
}{%
      \vspace{.25\baselineskip}%
    \end{minipage}%
  \end{Sbox}%
  \vspace{\baselineskip}%
  \begin{center}%
    \doublebox{\TheSbox}%
  \end{center}%
  \vspace{\baselineskip}%
}

\newcommand{\problemName}[1]{\noindent#1\vspace{.5\baselineskip}}
\newcommand{\InputTag}{\textsf{INPUT:}\ }
\newlength{\InputLength}
\newcommand{\problemInput}[1]{\hangindent=\InputLength\InputTag #1}     
\newcommand{\QuestionTag}{\textsf{QUESTION:}\ }
\newlength{\QuestionLength}
\newcommand{\problemQuestion}[1]{\hangindent=\QuestionLength\QuestionTag #1}

\newcommand{\problem}[3]{%
  \begin{Problem}\problemName{#1}\par\problemInput{#2}\par\problemQuestion{#3}\end{Problem}%
}

\let\Definition=\emph

\def\imagenes{}

\title{The star and biclique coloring and choosability problems}
\author{%
  Marina Groshaus\thanks{CONICET}~\thanks{Departamento de Computaci\'on, FCEN, Universidad de Buenos Aires, 
Buenos Aires, Argentina.} \and 
  Francisco J.\ Soulignac\footnotemark[1]~\footnotemark[2]~\thanks{Universidad Nacional de Quilmes, Buenos Aires, Argentina.} \and Pablo Terlisky\footnotemark[2]
}

\date{\normalsize\texttt{\{groshaus,fsoulign,terlisky\}@dc.uba.ar}}

\begin{document}
\maketitle
\begin{abstract}
  A biclique of a graph $G$ is an induced complete bipartite graph.  A star of $G$ is a biclique contained in the closed neighborhood of a vertex.  A star (biclique) $k$-coloring of $G$ is a $k$-coloring of $G$ that contains no monochromatic maximal stars (bicliques).  Similarly, for a list assignment $L$ of $G$, a star (biclique) $L$-coloring is an $L$-coloring of $G$ in which no maximal star (biclique) is monochromatic.  If $G$ admits a star (biclique) $L$-coloring for every $k$-list assignment $L$, then $G$ is said to be star (biclique) $k$-choosable.  In this article we study the computational complexity of the star and biclique coloring and choosability problems.  Specifically, we prove that the star (biclique) $k$-coloring and $k$-choosability problems are $\Sigma_2^p$-complete  and $\Pi_3^p$-complete for $k > 2$, respectively, even when the input graph contains no induced $C_4$ or $K_{k+2}$.   Then, we study all these problems in some related classes of graphs, including $H$-free graphs for every $H$ on three vertices, graphs with restricted diamonds, split graphs, threshold graphs, and net-free block graphs.

 \vspace*{.2\baselineskip} {\bf Keywords:} star coloring, biclique coloring, star choosability, biclique choosability.
\end{abstract}

\section{Introduction}

Coloring problems are among the most studied problems in algorithmic graph theory.  In its classical form, the $k$-coloring problem asks if there is an assignment of $k$ colors to the vertices of a graph in such a way that no edge is monochromatic.  Many generalizations and variations of the classical coloring problem have been defined over the years.  One of such variations is the clique $k$-coloring problem, in which the vertices are colored so that no maximal clique is monochromatic.  In this article we study the star and biclique coloring problems, which are variations of the coloring problem similar to clique colorings.  A biclique is a set of vertices that induce a complete bipartite graph $K_{n,m}$, while a star is a biclique inducing the graph $K_{1,m}$.  In the star (biclique) $k$-coloring problem, the goal is to color the vertices with $k$ colors without generating monochromatic maximal stars (bicliques).

The clique coloring problem has been investigated for a long time, and it is still receiving a lot of attention.  Recently, the clique $k$-coloring problem was proved to be \stp-complete~\cite{MarxTCS2011} for every $k \geq 2$, and it remains \stp-complete for $k=2$ even when the input is restricted to graphs with no odd holes~\cite{DefossezJGT2009}.  The problem has been studied on many other classes of input graphs, for which it is was proved to be \NP-complete or to require polynomial time (e.g.~\cite{BacsoGravierGyarfasPreissmannSebHoSJDM2004,CerioliKorenchendler2009,DefossezJGT2006,GravierHoangMaffrayDM2003,KratochvilTuzaJA2002,MacedoMachadoFigueiredo2012}).  Due to the close relation between cliques and bicliques, many problems on cliques have been translated in terms of bicliques (e.g.~\cite{AmilhastreVilaremJanssenDAM1998,PrisnerC2000a,TuzaC1984}).  However, there are some classical problems on cliques whose biclique versions were not studied until recently~\cite{EguiaSoulignacDMTCS2012,GroshausMonteroJoGT2012,GroshausSzwarcfiterGC2007,GroshausSzwarcfiterJGT2010}.  Clique colorings are examples of such problems; research on biclique colorings begun in 2010 in the Master Thesis of one of the authors~\cite{Terlisky2010} whose unpublished results are being extended in the present article.  It is worth mentioning that, despite its youthfulness, at least two articles on biclique colorings were written: \cite{MacedoMachadoFigueiredo2012} develop a polynomial time algorithm for biclique coloring some unichord-free graphs, and \cite{MacedoDantasMachadoFigueiredo2012} determines the minimum number of colors required by biclique colorings of powers of paths and cycles.

The list coloring problem is a generalization of the coloring problem in which every vertex $v$ is associated with a list $L(v)$, and the goal is to color each vertex $v$ with an element of $L(v)$ in such a way that no edge is monochromatic.  Function $L$ is called a list assignment, and it is a $k$-list assignment when $|L(v)| = k$ for every vertex $v$.  A graph $G$ is said to be $k$-choosable when it admits an $L$-coloring with no monochromatic edges, for every $k$-list assignment $L$.  The choosability problem asks whether a graph is $k$-choosable. In the same way as the coloring problem is generalized to the clique (star, biclique) coloring problem, the choosability problem is generalized to the clique (star, biclique) choosability problem.  That is, a graph $G$ is clique (star, biclique) $k$-choosable when it admits an $L$-coloring generating no monochromatic maximal cliques (star, bicliques), for every $k$-list assignment $L$.  The choosability problems seem harder than their coloring versions, because a universal quantifier on the list must be checked.  This difficulty is reflected for the $k$-choosability and clique $k$-choosability problems in the facts that the former is \ptwop-complete for every $k \geq 3$~\cite{GutnerTarsiDM2009}, whereas the latter is \ptp-complete for every $k \geq 2$~\cite{MarxTCS2011}.  In~\cite{MoharvSkrekovskiEJC1999} it is proven that every planar graph is clique $4$-choosable.  However, contrary to what happens with the clique coloring problem, there are not so many results regarding the complexity of the clique coloring problem for restricted classes of graphs.

In this paper we consider the star and biclique coloring and choosability problems, both for general graphs and for some restricted classes of graphs.  The star and biclique coloring and choosability problems are defined in Section~\ref{sec:preliminaries}, where we introduce the terminology that will be used throughout the article.  In Section~\ref{sec:general case}, we prove that the star $k$-coloring problem is \stp-complete for $k \geq 2$, and that it remains \stp-complete even when its input is restricted to $\{K_{2,2}, K_{k+2}\}$-free graphs.  Clearly, every maximal biclique of a $K_{2,2}$-free graph is a star.  Thus, we obtain as a corollary that the biclique $k$-coloring problem on $\{K_{2,2}, K_{k+2}\}$-free graphs is \stp-complete as well.  The completeness proof follows some of the ideas by Marx~\cite{MarxTCS2011} for the clique coloring problem.  In Section~\ref{sec:choosability} we show that the star $k$-choosability problem is \ptp-complete for $k \geq 2$, and that it remains \ptp-complete for $\{K_{2,2}, K_{k+2}\}$-free graphs.  Again, the \ptp-completeness of the biclique $k$-coloring problem on $\{K_{2,2}, K_{k+2}\}$-free is obtained as a corollary.  As in~\cite{MarxTCS2011}, we require a structure to force a color on a vertex.  The remaining sections of the article study the star and biclique coloring problems on graphs with restricted inputs.  These graphs are related to the graph $G$ that is generated to prove the \stp-completeness of the star coloring problem in Section~\ref{sec:general case}.  The aim is to understand what structural properties can help make the problem simpler.  In Section~\ref{sec:small forbiddens}, we discuss the star and biclique coloring and choosability problems on $K_3$-free, $P_3$-free, $\overline{P_3}$-free, and $\overline{K_3}$-free graphs.  For $K_3$, $P_3$ and $\overline{P_3}$, the star coloring and star choosability problems are almost trivial and can be solved in linear time.  On the other hand, both problems are as hard as they can be for $\overline{K_3}$-free graphs, even when the input is a co-bipartite graph.  In Section~\ref{sec:diamond-free} we prove that the star coloring problem is \NP-complete for diamond-free graphs and that the star choosability problem is \ptwop-complete for a superclass of diamond-free graphs.  If no induced $K_{i,i}$ is allowed for a fixed $i$, then the biclique coloring and the biclique choosability problems are also \NP-complete and \ptwop-complete.  In Section~\ref{sec:split}, the star coloring and the star choosability problems on split graphs are proved to be \NP-complete and \ptwop-complete, respectively.  Finally, Sections \ref{sec:threshold}~and~\ref{sec:block} show that the star coloring and the star choosability problems are equivalent for both threshold and net-free block graphs, and both can be solved in linear time.  Table~\ref{tab:results} sums up the results obtained in the article for the star coloring and star choosability problems.

\begin{table}
 \centering
 \begin{tabular}{|l|c|c|}
 \hline
 Graph class                                   & star $k$-coloring  & star $k$-choosability  \\
 \hline
 $\{K_{2,2}, K_{k+2}\}$-free                   & \stp-complete      & \ptp-complete         \\
 $K_3$-free, $P_3$-free, $\overline{P_3}$-free & $O(n+m)$           & $O(n+m)$              \\
 $\overline{K_3}$-free                         & \NP-complete       & \ptwop-complete       \\
 co-bipartite                                  &                    & \ptwop-complete       \\
 \{$W_4$, gem, dart\}-free                     & \NP-complete       & \ptwop-complete       \\
 diamond-free                                  & \NP-complete       & \ptwop                \\
 $\overline{C_4}$-free                         & \NP-complete       & \ptwop-complete       \\
 split                                         & \NP-complete       & \ptwop-complete       \\
 threshold                                     & $O(n+m)$           & $O(n+m)$              \\
 net-free block                                & $O(n+m)$           & $O(n+m)$              \\
 \hline
 \end{tabular}
 \caption{Complexity results obtained in this article.}\label{tab:results}
\end{table}

\section{Preliminaries}
\label{sec:preliminaries}

In this paper we work with simple graphs.  The vertex and edge sets of a graph $G$ are denoted by $V(G)$ and $E(G)$, respectively.  Write $vw$ to denote the edge of $G$ formed by vertices $v, w \in V(G)$.  For the sake of simplicity, $E(G)$ is also considered as the family of subsets of $V(G)$ containing the set $\{v,w\}$ for each $vw \in E(G)$.  For $v \in V(G)$, the \Definition{neighborhood} of $v$ is the set $N_G(v)$ of vertices adjacent to $v$, while the \Definition{closed neighborhood} of $v$ is $N_G[v] = N_G(v) \cup \{v\}$.  A vertex $w$ \Definition{dominates} $v$, and $v$ is \Definition{dominated} by $w$, when $N_G[v] \subseteq N_G[w]$, while $w$ \Definition{false dominates} $v$, and $v$ is \Definition{false dominated} by $w$, when $N_G(v) \subseteq N_G(w)$.  If $N_G[v] = N_G[w]$, then $v$ and $w$ are \Definition{twins}, and if $N_G(v) = N_G(w)$, then $v$ and $w$ are \Definition{false twins}.   The \Definition{degree} of $v$ is $d_G(v) = |N_G(v)|$.  A vertex is an \Definition{isolated vertex}, a \Definition{leaf}, and a \Definition{universal vertex} when $d_G(v)$ equals $0$, $1$, and $|V(G)|-1$, respectively.  We omit the subscripts from $N$ and $d$ when no ambiguities arise.

The \Definition{complement} of $G$ is the graph $\overline{G}$ where $V(\overline{G}) = V(G)$ and $E(\overline{G}) = \{vw \mid vw \not\in E(G)\}$.  For a graph $H$, the \Definition{union} of $G$ and $H$ is the graph $G \cup H$ where $V(G \cup H) = V(G) \cup V(H)$ and $E(G \cup H) = E(G) \cup E(H)$.  Write $G = H$ to indicate that $G$ and $H$ are isomorphic.  The \Definition{$n$-cycle} graph ($n \geq 3$), denoted by $C_n$, is the connected graph that has $n$ vertices of degree $2$.  The \Definition{$n$-path} graph, denoted by $P_n$, is the graph obtained from $C_n$ by removing an edge.  The \emph{$n$-wheel} graph ($n \geq 3$), denoted by $W_n$, is the graph obtained from $C_n$ by inserting a universal vertex.  The \Definition{diamond}, \Definition{gem}, \Definition{dart}, and \Definition{net} are the graphs shown in Figure~\ref{fig:diam-dart-gem}.  The \Definition{$n$-complete} graph, denoted by $K_n$, is the graph formed by $n$ pairwise adjacent vertices.  An \Definition{independent set} is a subset of $V(G)$ formed by pairwise non-adjacent vertices.  Graph $G$ is \Definition{bipartite} when $V(G)$ can be partitioned into two independent sets $S$ and $T$. In this case, the unordered pair $ST$ is called a \Definition{bipartition} of $G$.  The \Definition{$(n,m)$-complete bipartite} graph ($n \geq 1$, $m \geq 1$), denoted by $K_{n,m}$, is the graph isomorphic to $\overline{K_n \cup K_m}$.  Note that $K_{n,m}$ is a bipartite graph.  The graph $K_{1,n}$ is also called the \Definition{$n$-star} graph.  The universal vertices of $K_{1,n}$ are referred to as the \Definition{centers} of $K_{1,n}$, while $K_{1,n}$ is said to be \Definition{centered} at $v$.  Note that $K_{1,1}$ has two centers.  

\begin{figure}
 \hfill\includegraphics{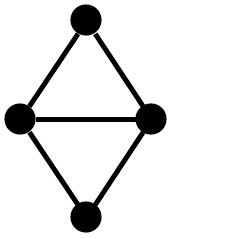}\hfill\includegraphics{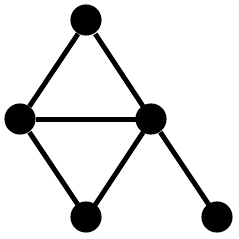}\hfill\includegraphics{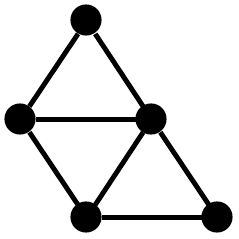}\hfill{}\includegraphics{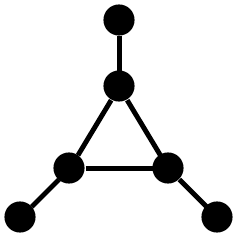}\hfill{}
 \caption{The diamond, dart, gem, and net graphs are shown from left to right.}\label{fig:diam-dart-gem}
\end{figure}

Let $W \subseteq V(G)$.  We write $G[W]$ to denote the subgraph of $G$ induced by $W$, and $G \setminus W$ to denote $G[V(G) \setminus W]$.  Set $W$ is said to be a \Definition{clique}, \Definition{biclique}, and \Definition{star} when $G[W]$ is isomorphic to a complete, bipartite complete, and star graph, respectively.  For the sake of simplicity, we use the terms \Definition{clique} and \Definition{biclique} to refer to $G[W]$ as well.  Moreover, we may refer to $ST$ as a biclique or star when $G[S \cup T]$ is a bipartite complete or star graph with bipartition $ST$. The family of maximal cliques, maximal bicliques, and maximal stars are denoted by $\MC(G)$, $\MB(G)$, and $\MS(G)$, respectively. 

A sequence of distinct vertices $P = v_1, \ldots, v_n$ is a \Definition{path} of \Definition{length} $n-1$ when $v_i$ is adjacent to $v_{i+1}$.  If in addition $v_{n}$ is adjacent to $v_1$, then $P$ is a \Definition{cycle} of \Definition{length} $n$.  A \Definition{tree} is a connected graph that contains no cycles.  A \Definition{rooted tree} is a tree $T$ with a fixed vertex $r \in V(T)$ called the \Definition{root} of $T$.  The \Definition{parent} of $v$ in $T$ is the neighbor of $v$ in its path to $r$.  A path (resp.\ cycle) of $G$ is \Definition{chordless} when $G[P] = P_n$ (resp.\ $G[P] = C_n$).  A \Definition{hole} is a chordless cycle of length at least $4$.  A graph $G$ is said to be \Definition{$H$-free}, for some graph $H$, when no induced subgraph of $G$ is isomorphic to $H$.  Similarly, $G$ is $\mathcal{F}$-free, for any family of graph $\mathcal{F}$, when $G$ is $H$-free for every $H \in \mathcal{F}$.  A graph is \Definition{chordal} when it is $\{C_n\}_{n\geq4}$-free, i.e., chordal graphs have no holes.  

A \Definition{coloring} of $G$ is a function $\rho$ that maps each vertex $v \in V(G)$ to a color $\rho(v) \in \mathbb{N}$.  When $\rho(v) \leq k$ for every $v \in V(G)$, $\rho$ is called a \Definition{$k$-coloring}.  We define $\rho(W) = \{\rho(v) \mid v \in W\}$ for any $W \subseteq V(G)$.  Set $W$ is said to be \Definition{$\rho$-monochromatic} when $|\rho(W)| = 1$.  When there is no ambiguity, we say that $W$ \Definition{monochromatic} instead of $\rho$-monochromatic.  For a family $\mathcal{F}$ of subsets of $V(G)$, we say that $\rho$ is a \Definition{proper coloring} of $\mathcal{F}$ if no $W \in \mathcal{F}$ is monochromatic.  Four kinds of families are considered in this article.  A coloring $\rho$ is a \Definition{vertex}, \Definition{clique}, \Definition{biclique}, and \Definition{star coloring} when $\rho$ is a proper coloring of $E(G)$, $\MC(G)$, $\MB(G)$, and $\MS(G)$, respectively.  The problems of finding a proper coloring for these families are defined as follows.

\begin{Problem}
  \problemName{\textsc{Vertex (resp.\ clique, biclique, star) $k$-coloring}}

  \problemInput{A connected graph $G$ and a value $k \in \mathbb{N}$.}

  \problemQuestion{Is there a vertex (resp.\ clique, biclique, star) $k$-coloring of $G$?}
\end{Problem}

List colorings are a generalization of colorings. A \Definition{list assignment} of $G$ is a function that maps each vertex $v \in V(G)$ to a set $L(v) \subseteq \mathbb{N}$.  When $|L(v)| = k$ for every $v \in V(G)$, $L$ is called a \Definition{$k$-list assignment}.  An $L$-coloring of $G$ is a coloring $\rho$ such that $\rho(v) \in L(v)$ for every $v \in V(G)$.  Define $L(W) = \bigcup\{L(v) \mid v \in W\}$ for any $W \subseteq V(G)$.  Given family $\mathcal{F}$ of subset of $V(G)$ and a number $k \in \mathbb{N}$, graph $G$ is said to be \Definition{$k$-choosable} with respect to $\mathcal{F}$ when there exists a proper $L$-coloring of $\mathcal{F}$ for every $k$-list assignment $L$ of $G$.  Graph $G$ is \Definition{vertex}, \Definition{clique}, \Definition{biclique}, and \Definition{star $k$-choosable} when $G$ is $k$-choosable with respect to $E(G)$, $\MC(G)$, $\MB(G)$, and $\MS(G)$, respectively.  The problem of determining if $G$ is $k$-choosable is defined as follows.

\begin{Problem}
  \problemName{\textsc{Vertex (resp.\ clique, biclique, star) $k$-choosability}}

  \problemInput{A connected graph $G$ and a value $k \in \mathbb{N}$.}

  \problemQuestion{Is $G$ vertex (resp.\ clique, biclique, star) $k$-choosable?}
\end{Problem}

The vertex (resp.\ clique, biclique, star) \Definition{chromatic number}, denoted by $\chi(G)$ (resp.\ $\chi_C(G)$, $\chi_B(G)$, and $\chi_S(G)$), is the minimum $k \in \mathbb{N}$ such that $G$ admits a vertex (resp.\ clique, biclique, star coloring) $k$-coloring. Similarly, the vertex (clique, biclique, star) \Definition{choice number}, denoted by $ch(G)$ (resp.\ $ch_C(G)$, $ch_B(G)$, $ch_S(G)$) is the minimum number $k \in \mathbb{N}$ such that $G$ is vertex (resp.\ clique, biclique, star coloring) $k$-choosable.  By definition, $\chi(G) \leq ch(G)$ and $\chi_*(G) \leq ch_*(G)$ for $* \in \{C, B, S\}$.  

For a function $f$ with domain $D$, the \Definition{restriction} of $f$ to $D' \subseteq D$ is the function $f'$ with domain $D'$ where $f'(x) = f(x)$ for $x \in D'$.  In such case, $f$ is said to be an \Definition{extension} of $f'$ to $D$.  A \Definition{leafed vertex} is a vertex adjacent to a leaf.  For the sake of simplicity, whenever we state that $G$ \emph{contains a leafed vertex} $v$, we mean that $G$ contains $v$ and a leaf adjacent to $v$.  It is well known that $\chi(G) \geq \chi(H)$ for every induced subgraph $H$ of $G$.  Such a property is false for clique, biclique, and star colorings.  In particular, for any graph $H$, it is possible to build a graph $G$ such that $ch_*(G) = 2$ and $G$ contains $H$ as an induced subgraph.  For $* = C$, graph $G$ is built from $H$ by iteratively inserting a twin of each vertex of $H$.  Similarly, $G$ is obtained by inserting false twins for $*  = B$, while, by the next observation, $G$ is obtained by inserting a leaf adjacent to each vertex for $* = S$.  

\begin{observation}\label{obs:leafed vertex}
 Let $G$ be a graph with a list assignment $L$, $v$ be a leafed vertex of $G$, and $l$ be a leaf adjacent to $v$.  Then, any $L$-coloring of $G \setminus l$ can be extended into an $L$-coloring of $G$ in such a way that there is no monochromatic maximal star with center in $v$.
\end{observation}

A \Definition{block} is a maximal set of twin vertices.  If $v$ and $w$ are twin vertices, then $\{v\}\{w\}$ is both a maximal star and a maximal biclique, and thus $v$ and $w$ have different colors in any star or biclique $L$-coloring $\rho$.  Consequently, $|\rho(W)| = W$ for any block $W$ of $G$.  We record this fact in the following observation.

\begin{observation}\label{obs:block coloring}
 Let $G$ be a graph with a list assignment $L$, and $v,w$ be twin vertices.  Then, $\rho(v) \neq \rho(w)$ for any star or biclique $L$-coloring $\rho$ of $G$.
\end{observation}

\section{Complexity of star and biclique coloring}
\label{sec:general case}

In this section we establish the hardness of the star and biclique coloring problems by showing that both problems are \stp-complete.  The main result of this section is that \stcol{k} is \stp-complete for every $k \geq 2$, even when its input is restricted to $\{C_4, K_{k+2}\}$-free graphs.  Since all the bicliques of a $C_4$-free graph are stars, this immediately implies that \bccol{k} is also \stp-complete for $\{C_4, K_{k+2}\}$-free graphs.  The hardness results is obtained by reducing instances of the \qsat{2} problem.  The \qsat{h} problem is known to be $\Sigma_h^p$-complete for every $h$~\cite{Papadimitriou1994}, and is defined as follows.

\problem{\textbf{Quantified $3$-satisfiability with $h$ alternations} (\qsat{h})}
{A formula $\phi(\VEC x_1, \VEC x_2, \ldots, \VEC x_h)$ that is in $3$-CNF if $h$ is odd, while it is in $3$-DNF if $h$ is even.}
{Is $(\exists\VEC x_1)(\forall\VEC x_2)(\exists\VEC x_3)\ldots (Q_h\VEC x_h)\phi(\VEC x_1, \VEC x_2, \VEC x_3, \ldots, \VEC x_h)$ true? ($Q_h \in \{\exists, \forall\}$.)}

Recall that $\phi$ is in $3$-CNF if it is a conjunction of clauses where each clause is a disjunction with three literals.  Similarly, $\phi$ is in $3$-DNF when it is a disjunction of clauses, each clause being a conjunction with three literals.

\subsection{Keepers, switchers, and clusters}

In this section we introduce the keeper, switcher, and cluster connections that are required for the reductions.  The keeper connections are used to force the same color on a pair of vertices, in any star coloring.  Conversely, the switcher connections force some vertices to have different colors.  Finally, the cluster connections are used to represent the variables of a DNF formula.  We begin defining the keeper connections.

\begin{defn}[$k$-keeper]\label{def:k-keeper connection}
  Let $G$ be a graph and $v, w \in V(G)$.  Say that $K \subset V(G)$ is a \Definition{$k$-keeper connecting $v, w$} ($k \geq 2$) when $K$ can be partitioned into a clique $D = \{d_1, \ldots, d_{k-1}\}$ and $k-1$ cliques $C_1, \ldots, C_{k-1}$ with $k$ vertices each in such a way that $D \cup \{v,w\}$ and $C_{i} \cup \{d_i\}$ are cliques for \range{i}{1}{k-1}, and there are no more edges incident to vertices in $K$.
\end{defn}

Figure~\ref{fig:keeper-switcher}~(a) shows a $k$-keeper connecting two vertices $v$ and $w$. The main properties of $k$-keepers are summarized in the following lemma.

\begin{lemma}\label{lem:keeper properties}
 Let $G$ be a graph and $K$ be a $k$-keeper connecting $v,w \in V(G)$ ($k \geq 2$). Then,
\begin{enumerate}[(i)]
  \item no induced hole or $K_{k+2}$ contains a vertex in $K$,\label{lem:keeper properties:forbidden}
  \item $v$ and $w$ have the same color in any star $k$-coloring of $G$, and\label{lem:keeper properties:colors}
  \item Any $k$-coloring $\rho$ of $G \setminus K$ in which $\rho(v) = \rho(w)$ can be extended into a $k$-coloring of $G$ in such a way that no monochromatic maximal star contains a vertex in $K$.\label{lem:keeper properties:extension}
\end{enumerate}
\end{lemma}

\begin{proof}
Let $C_1$, \ldots, $C_{k-1}$ and $D = \{d_1, \ldots, d_{k-1}\}$ be as in Definition~\ref{def:k-keeper connection}.  (\ref{lem:keeper properties:forbidden}) is trivial.  

(\ref{lem:keeper properties:colors}) Let $\rho$ be a star $k$-coloring of $G$ and fix $\range{i,j}{1}{k-1}$.  Since $C_i$ is a block of size $k$, it contains vertices $c_i, c_j, c_v$ with colors $\rho(d_i)$, $\rho(d_j)$, and $\rho(v)$, respectively, by Observation~\ref{obs:block coloring}.  Hence, $\rho(d_i) \neq \rho(v)$ since otherwise $\{d_i\}\{v, c_v\}$ would be a monochromatic maximal star, and, similarly, $\rho(d_i) \neq \rho(d_j)$ because $\{d_i\}\{d_j, c_j\}$ is not monochromatic. In other words, $|\rho(D)| = k-1$ and $\rho(v) \not\in \rho(D)$.  Replacing $v$ and $w$ in the reasoning above, we can conclude that $\rho(w) \not\in \rho(D)$ as well.  Therefore, $\rho(v) = \rho(w)$.

(\ref{lem:keeper properties:extension}) To extend $\rho$, set $\rho(D) = \{1, \ldots, k\} \setminus \{\rho(v)\}$, and $\rho(C_i) = \{1, \ldots, k\}$ for every \range{i}{1}{k-1}.  It is not hard to see that no maximal star with a vertex in $K$ is monochromatic.
\end{proof}

\begin{figure}
 \centering
 \begin{tabular}{ccc}
 \includegraphics{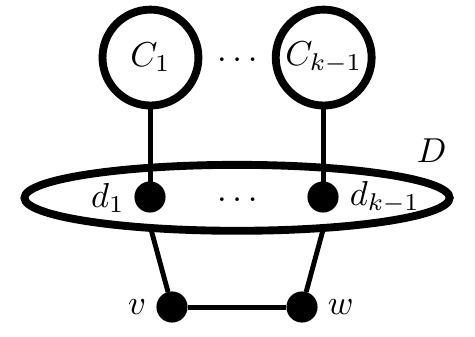} & \includegraphics{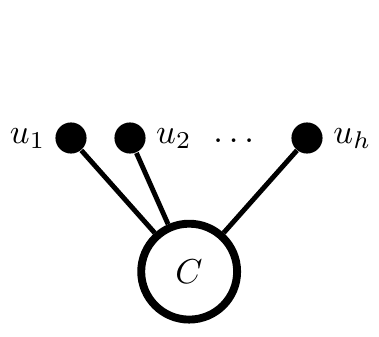} & \includegraphics{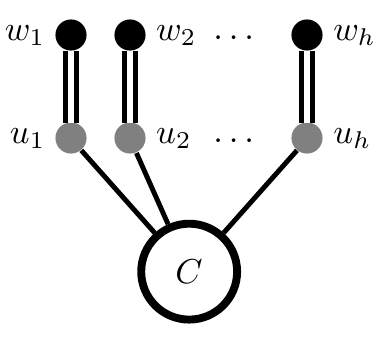}\\
 (a) & (b) & (c)
 \end{tabular}
 \caption{(a) A $k$-keeper connecting $v,w$.  (b) A $k$-switcher connecting $\{u_1, \ldots, u_h\}$. (c) A long $k$-switcher connecting $\{w_1, \ldots, w_h\}$.  In every figure, gray nodes represent leafed vertices, circular shapes represent cliques, and double lines represent a $k$-keeper connecting two vertices.}\label{fig:keeper-switcher}
\end{figure}

The \Definition{switcher connection}, whose purpose is to force a set of vertices to have at least two colors, is now defined (see Figure~\ref{fig:keeper-switcher}~(b)).

\begin{defn}[$k$-switcher]\label{def:switcher}
  Let $G$ be a graph and $U = \{u_1, \ldots, u_h\}$ be an independent set of $G$ with $h \geq 2$.  Say that $C \subset V(G)$ is a \Definition{$k$-switcher connecting $U$} ($k \geq 2$) when $|C| = k$, $C \cup \{u_i\}$ is a clique for \range{i}{1}{h}, and there are no more edges incident to vertices in $C$.
\end{defn}

The following result is the analogous of Lemma~\ref{lem:keeper properties} for switchers.

\begin{lemma}\label{lem:switcher properties}
 Let $G$ be a graph and $C$ be a $k$-switcher connecting $U \subset V(G)$ ($k \geq 2$). Then,
\begin{enumerate}[(i)]
  \item $|\rho(U)| \geq 2$ for any star $k$-coloring $\rho$ of $G$, and \label{lem:switcher properties:colors}
  \item Any $k$-coloring $\rho$ of $G \setminus S$ in which $|\rho(U)| \geq 2$ can be extended into a $k$-coloring of $G$ in such a way that no monochromatic maximal star has its center in $C$.\label{lem:switcher properties:extension}
\end{enumerate}
\end{lemma}

As defined, switchers are not useful for proving the hardness of the star coloring problem for $C_4$-free graphs.  The reason is that their connected vertices must have no common neighbors to avoid induced $C_4$'s, and our proof requires vertices with different colors and common neighbors.  To solve this problem we extend switchers into long switchers by combining them with keepers (see Figure~\ref{fig:keeper-switcher}~(c)).  We emphasize that the set of vertices connected by a long switcher need not be an independent set.

\begin{defn}[long $k$-switcher]\label{def:long switcher}
  Let $G$ be a graph, and $W = \{w_1, \ldots, w_h\}$ be a set of vertices of $G$ with $h \geq 2$.  Say that $S \subset V(G)$ is a \Definition{long $k$-switcher connecting $W$} ($k \geq 2$) when $S$ can be partitioned into an independent set of leafed vertices $U = \{u_1, \ldots, u_h\}$, a $k$-switcher $C$, and $k$-keepers $Q_1, \ldots, Q_h$ in such a way that $C$ connects $U$, $Q_i$ connects $w_i, u_i$ for \range{i}{1}{h}, and there are no more edges adjacent to vertices in $U$.
\end{defn}

The analogous of Lemma~\ref{lem:switcher properties} follows; its proof is a direct consequence of Observation~\ref{obs:leafed vertex} and Lemmas \ref{lem:keeper properties}~and~\ref{lem:switcher properties}.

\begin{lemma}\label{lem:long switcher properties}
 Let $G$ be a graph and $S$ be a long $k$-switcher connecting $W \subset V(G)$ ($k \geq 2$). Then,
\begin{enumerate}[(i)]
  \item no induced $C_4$ or $K_{k+2}$ of $G$ contains a vertex of $S$,
  \item $|\rho(W)| \geq 2$ for any star $k$-coloring $\rho$ of $G$, and \label{lem:long switcher properties:colors}
  \item Any $k$-coloring $\rho$ of $G \setminus S$ in which $|\rho(W)| \geq 2$ can be extended into a $k$-coloring of $G$ in such a way that no monochromatic maximal star has its center in $S$.\label{lem:long switcher properties:extension}
\end{enumerate}
\end{lemma}

The last type of connection that we require is the cluster, which is used to represent variables of DNF formulas.  As switchers, clusters are used to connect several vertices at the same time.  Specifically, clusters require two sets $X$, $-X$, and a vertex $s$.  The purpose of the connection is to encode all the valuations of the variables in a formula $\phi$ by using monochromatic stars with center in $s$.  Thus, if a variable $\VP x$ is being represented by $X$ and $-X$, then each monochromatic maximal star with center in $s$ contains all the vertices in $X$ and none of $-X$, when $\VP x$ is true, or it contains all the vertices of $-X$ and none of $X$, when $\VP x$ is false.

\begin{defn}[$\ell$-cluster]\label{def:cluster}
  Let $G$ be a graph with vertices $s$, $X = \{x_1, \ldots, x_\ell\}$, and $-X = \{-x_1, \ldots, -x_\ell\}$ ($\ell \geq 2$).  Say that $K \subseteq V(G)$ is an \Definition{$\ell$-cluster connecting $\langle s, X, -X\rangle$} when $K$ has $\ell$ leafed vertices $v_1, \ldots, v_\ell$, vertex $s$ is adjacent to all the vertices in $X \cup -X \cup K$, sequence $x_1, -x_1, v_1, \ldots, x_\ell, -x_\ell, v_\ell$ is a hole, and there are no more edges incident to vertices in $K$.  For \range{i}{1}{\ell}, we write $X[i]$ and $-X[i]$ to refer to $x_i$ and $-x_i$, respectively.  
\end{defn}

Note that if $K$ is an $\ell$-cluster connecting $\langle s, X, -X\rangle$, then the subgraph induced by $K \cup X \cup -X \cup \{s\}$ is isomorphic to a $3\ell$-wheel that has $s$ as its universal vertex.  As mentioned, the main property of clusters is that its members can be colored in such a way that monochromatic stars represent valuations.

\begin{lemma}\label{lem:cluster properties}
 Let $G$ be a graph and $K$ be an $\ell$-cluster connecting $\langle s, X, -X\rangle$ for $\{s\} \cup X \cup -X \subseteq V(G)$.  Then, 
 \begin{enumerate}[(i)]
  \item if $N(x) \cap N(x') \subseteq K \cup \{s\}$ for every $x, x' \in X \cup -X$, then no induced $K_4$ or $C_4$ of $G$ contains a vertex in $K$,
  \item if $\{s\}S$ is a maximal star of $G$ with $S \cap K = \emptyset$, then $S \cap (X \cup -X)$ equals either $X$ or $-X$, and \label{lem:cluster properties:stars}
  \item any $k$-coloring $\rho$ of $G \setminus K$ ($k \geq 2$) can be extended into a $k$-coloring of $G$ in such a way that $\rho(s) \not\in \rho(K)$, and no monochromatic maximal star of $G$ has its center in $K$. \label{lem:cluster properties:extension}
 \end{enumerate}
\end{lemma}

\begin{proof}
  Let $K = \{v_1, \ldots, v_\ell\}$, $x_1, \ldots, x_{\ell}$, and $-x_1, \ldots, -x_{\ell}$ be as in Definition~\ref{def:cluster}.
  
  (i) For \range{i}{1}{\ell}, the non-leaf neighbors of $v_i$ are $s, -x_i, x_{i+1}$.  Since $-x_i$ and $x_{i+1}$ are not adjacent and $N(-x_i) \cap N(x_{i+1}) = \{v_i, s\}$, it follows that $v_i$ belongs to no induced $K_4$ or $C_4$ of $G$.  
  
  (ii) Suppose $S \cap K = \emptyset$.  If $x_i \in S$ for some \range{i}{1}{\ell}, then $-x_{i} \not\in S$ because $-x_i \in N(x_i)$.  Then, since $v_i \not\in S$, it follows that $x_{i+1}$ belongs to $S$.  Consequently, by induction on $i$, we obtain that $S \cap X$ equals either $X$ or $\emptyset$.  In the former case, $S \cap -X = \emptyset$ because every vertex in $-X$ has a neighbor in $X$.  In the latter case, $S \cap -X = -X$ because $S \cap K = \emptyset$ and at least one of $\{x_{i+1}, -x_i, v_i\}$ belongs to $S$ for every \range{i}{1}{\ell}.
  
  (iii)  Just extend $\rho$ so that $\rho(s) \not\in \rho(K)$, and color the leaves according to Observation~\ref{obs:leafed vertex}.
\end{proof}

\subsection{Hardness of the star-coloring problem}
\label{sec:general case:2-stp}

It is well known that a problem $P$ is \stp when the problem of authenticating a positive certificate of $P$ is $\coNP$~\cite{Papadimitriou1994}.  For the star $k$-coloring problem, a $k$-coloring of $G$ can be taken as a positive certificate.  Since it is $\coNP$ to authenticate that a $k$-coloring of $G$ is indeed a star coloring, we obtain that \stcol{k} is \stp.  The following theorem states the \stp-hardness of the problem.  

\begin{theorem}\label{thm:stcol-2}
 \stcol{k} is \stp-complete and it remains \stp-complete even when its input is restricted to $\{C_4, K_{k+2}\}$-free graphs.
\end{theorem}

\begin{proof}
  We already know that $\stcol{k}$ belongs to \stp.  To prove its hardness, we show a polynomial time reduction from \qsat{2}.  That is, for any $3$-DNF formula $\phi(\VEC x, \VEC y)$ with $\ell \geq 2$ clauses $\VP P_1, \ldots, \VP P_\ell$, and $n+m$ variables $\VEC x= \VP x_1, \ldots, \VP x_n$, $\VEC y = \VP y_1, \ldots, \VP y_m$, we build a graph $G$ that admits a star $k$-coloring if and only if $(\exists\VEC x)(\forall \VEC y)\phi(\VEC x, \VEC y)$ is true.  For the sake of simplicity, in this proof we use $i$, $j$, $h$, and $q$ as indices that refer to values in $\{1, \ldots, n\}$, $\{1, \ldots, m\}$, $\{1, \ldots, \ell\}$, and $\{1, \ldots, k\}$.

  Graph $G$ can be divided into connection, inner, and leaf vertices.  \Definition{Connection} vertices are in turn divided into a set of \Definition{clause} vertices $P = \{p_1, \ldots, p_\ell\}$, a set of \Definition{$x$-vertices} $X = \{x_1, \ldots, x_n\}$, two sets of \Definition{$x$-literal} vertices $X_i$ and $-X_i$ with $\ell$ vertices each (for each $i$), two sets of \Definition{$y$-literal} vertices $Y_j$ and $-Y_j$ with $\ell$ vertices each (for each $j$), a set of \Definition{color} vertices $C = \{c_1, \ldots, c_k\}$, and two special vertices $s$ and $t$. Let $L_X = \bigcup_i(X_i \cup -X_i)$, and $L_Y = \bigcup_j^m(Y_j \cup -Y_j)$.  \Definition{Inner} vertices are those vertices included in switchers, keepers, and clusters of $G$.  Inner vertices and the edges between the connection vertices are given by the next rules.
\begin{description}
  \item [Edges:]  $s$ is adjacent to all the vertices in $P$, and if $\VP x_i$ (resp.\ $\overline{\VP x_i}$, $\VP y_j$, $\overline{\VP y_j}$) is a literal of $\VP P_h$, then $p_h$ is adjacent $-X_i[h]$ (resp.\ $X_i[h]$, $-Y_j[h]$, $Y_j[h]$).  

  \item [Keepers:] there is a $k$-keeper connecting $s$ with $t$.

  \item [Long switchers:] there are long $k$-switchers connecting $\{x_i, -X_i[h]\}$ and $\{X_i[h], -X_i[h]\}$ (for every $i,h$), $\{c_1, p_h\}$ (for every $h$), $\{c_2, y\}$ for every $y \in L_Y$, $\{c_q, w\}$ for every $q > 2$ and every connection vertex $w \neq s$, $\{c_1, c_2\}$, and $\{c_2,t\}$.

  \item [Variables:] there are $\ell$-clusters connecting $\langle s, X_i, -X_i\rangle$ and $\langle s, Y_j, -Y_j\rangle$.
\end{description}
  Finally, each connection vertex other than $s$ is leafed.  This ends up the construction of $G$, which can be easily computed from $\phi(\VEC x, \VEC y)$ in polynomial time. Figure~\ref{fig:coloring2} depicts a schema of the graph.

\begin{figure}
  \centering
  \includegraphics{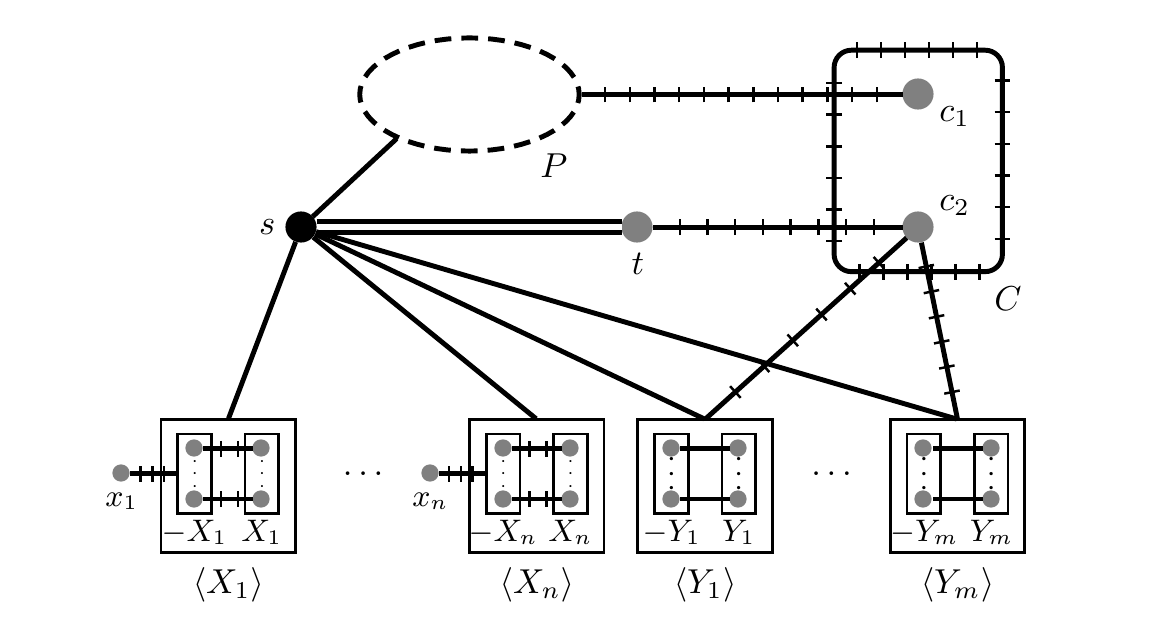}
  \caption{Schema of the graph obtained from $\phi$ in Theorem~\ref{thm:stcol-2}.  For the sake of simplicity, we omit $c_3, \ldots, c_k$ and the edges from $P$ to $L_X \cup L_Y$.  Circular shapes drawn with dashes represent independent sets; marked edges between two vertices represent $k$-switchers connecting them; circular shapes with marks represent sets of vertices pairwise connected by $k$-switchers; and squares represent sets of vertices.}\label{fig:coloring2}
\end{figure}

  Before dealing with the star $k$-coloring problem on $G$, we show that $G$ is $\{C_4, K_{k+2}\}$-free.  By statement~(\ref{lem:keeper properties:forbidden}) of Lemmas \ref{lem:keeper properties}, \ref{lem:long switcher properties}~and~\ref{lem:cluster properties}, it suffices to prove that the subgraph $H$ induced by the connection vertices is $\{K_{k+2}, C_4\}$-free.  For this, observe that any induced $C_4$ or $K_{k+2}$ must contain a vertex of $L_X \cup L_Y$ because $\{s\}P$ is an induced star of $H$ and $x_i$, $c_q$, and $t$ have degree at most $1$ in $H$.   Now, $X_i[h]$ has at most three neighbors in $H$, namely $-X_i[h]$, $s$, and maybe $p_h$.  Hence, since $N(-X_i[h]) = \{s, X_i[h]\}$ when $X_i[h]$ is adjacent to $p_h$, we obtain that $X_i[h]$ belongs to no induced $K_{k+2}$ nor $C_4$.  A similar analysis is enough to conclude no vertex of $L_X \cup L_Y$ belongs to an induced $K_{k+2}$ nor $C_4$, thus $H$ is $\{C_4, K_{k+2}\}$-free.

  Now we complete the proof by showing that $(\exists\VEC x)(\forall\VEC y)\phi(\VEC x, \VEC y)$ is true if and only if $G$ admits a star $k$-coloring.  Suppose first that $(\forall\VEC y)\phi(\VEC x, \VEC y)$ is true for some valuation $\nu:\VEC x \to \{0,1\}$, and define $\rho$ as the $k$-coloring of $G$ that is obtained by the following two steps.  First, set $\rho(c_q) = q$, $\rho(X_i) = \rho(x_i) = 2-\nu(\VP x_i)$, $\rho(-X_i) = 1 + \nu(\VP x_i)$, $\rho(p_k) = 2$, and $\rho(L_Y) = \rho(s) = \rho(t) = 1$.  Next, iteratively set $\rho$ for the leaves and inner vertices according to Observation~\ref{obs:leafed vertex}, and statement~(\ref{lem:keeper properties:extension}) of Lemmas~\ref{lem:keeper properties}, \ref{lem:long switcher properties},~and~\ref{lem:cluster properties}.  Observe that the second step is well defined because every pair of vertices connected by $k$-keepers have the same color, while every pair of vertices connected by long $k$-switchers have different colors.

  We claim that $\rho$ is a star $k$-coloring of $G$.  To see why, consider a maximal star $\{w\}S$ of $G$ and observe that $w$ cannot be a leaf.  If $w \neq s$, then $\{w\}S$ is not monochromatic by Observation~\ref{obs:leafed vertex} and Lemmas~\ref{lem:keeper properties}, \ref{lem:long switcher properties}~and~\ref{lem:cluster properties}.  Suppose, then, that $w = s$ and, moreover, that $S\setminus P$ is monochromatic.  Then, no $\ell$-cluster intersects $S$ by statement (\ref{lem:cluster properties:extension}) of Lemma~\ref{lem:cluster properties}.  Consequently, by statement (\ref{lem:cluster properties:stars}) of Lemma~\ref{lem:cluster properties}, $S \cap (X_i \cup -X_i)$ equals either $X_i$ or $-X_i$, while $S \cap (Y_j \cup -Y_j)$ equals either $Y_j$ or $-Y_j$ for every $i$ and every $j$.  Extend $\nu$ to include $\VEC y$ in its domain, so that $\nu(\VP y_j) = 1$ if and only if $Y_j \subseteq S$.  By hypothesis, $\nu(\phi(\VEC x, \VEC y)) = 1$, thus there is some clause $\VP P_h$ whose literals are all true according to $\nu$.  If $p_h$ has some neighbor in $-X_i$, then $\nu(\VP x_i) = 1$, thus $-X_i \not\subset S$ because $\rho(-X_i) = 2$.  Similarly, if $p_h$ has some neighbor in $X_i$ (resp.\ $-Y_j$, $Y_j$), then $X_i \not\subset S$ (resp.\ $-Y_j \not\subset S$, $Y_j\not\subset S$).  Therefore, since $P$ is an independent set and $S \subset \{t\} \cup L_X \cup L_Y \cup P$, we obtain that $p_k \in S$, thus $\{s\}S$ is not monochromatic.

  For the converse, let $\rho$ be a star $k$-coloring of $G$.  Since there is are long $k$-switchers connecting $\{c_1, c_2\}$ and $\{c_q, w\}$ for every $q > 2$ and every connection vertex $w \neq s$, we obtain that $|\rho(C)| = k$ and $\rho(w) \in \{\rho(c_1), \rho(c_2)\}$ by statement~(\ref{lem:long switcher properties:colors}) of Lemma~\ref{lem:long switcher properties}. Define $\nu\colon \VEC x + \VEC y \to \{0,1\}$ as any valuation in which $\nu(\VP x_i) = 1$ if and only if $\rho(x_i) = \rho(c_1)$.  Since $\nu(\VP y_j)$ can take any value from $\{0,1\}$, it is enough to prove that $\nu(\phi(\VEC x, \VEC y)) = 1$.  Let $V_X = \bigcup_i((X_i \mid \nu(\VP x_i) = 1) \cup (-X_i \mid \nu(\VP x_i) = 0))$ and $V_Y = \bigcup_j((Y_j \mid \nu(\VP y_j) = 1) \cup (-Y_j \mid \nu(\VP y_j) = 0))$.

  Let $S = V_X \cup V_Y \cup \{t\}$.  By construction, $S$ is an independent set, thus $\{s\}S$ is a star of $G$.  Recall that there are long $k$-switchers connecting $\{x_i, -X_i[h]\}$ and $\{X_i[h], -X_i[h]\}$.  Hence $\rho(x_i) = \rho(X_i)$ and $\rho(-X_i) \neq \rho(x_i)$ by statement~(\ref{lem:long switcher properties:colors}) of Lemma~\ref{lem:long switcher properties}.  This implies that $\rho(V_X) = \{1\}$.  Similarly, there are long $k$-switchers connecting $\{c_2, y\}$ for each $y \in L_Y$, thus $\rho(V_Y) \subset \rho(L_Y) = 1$ as well.  Finally, using Lemma~\ref{lem:keeper properties}, we obtain that $\rho(s) = \rho(t) = 1$ because there is a long $k$-switcher connecting $\{c_2,t\}$ and a $k$-keeper connecting $t, s$.  So, by hypothesis, $\{s\}S$ is not a maximal star, which implies that $S \cup \{w\}$ is also an independent set for some $w \in N(s)$.

  Since $t$ and $s$ have the same neighbors in the $k$-keeper $K$ connecting them, it follows that $w \not\in K$.  Similarly, all the vertices in a cluster are adjacent to at least one vertex of $V_X \cup V_Y$.  Finally, each vertex of $L_X \cup L_Y$ either belongs or has some neighbor in $V_X \cup V_Y$.  Consequently, $w = p_h$, i.e., $w$ represents some clause $\VP P_h$.  If $\VP x_i$ is a literal of $\VP P_h$, then $p_h$ is adjacent to $-X_i[h]$.  Hence, since $V_X \cup V_Y \cup \{p_h\}$ is an independent set, it follows that $-X_i[h] \not\in V_X$.  By the way $V_X$ is defined, this means that $\nu(\VP x_i) = 1$.  Similar arguments can be used to conclude that if $\VP l$ is a literal of $\VP P_h$, then $\nu(\VP l) = 1$.  That is, $\VP P_h$ is satisfied by $\nu$, thus $(\exists\VEC x)(\forall \VEC y)\phi(\VEC x, \VEC y)$ is true.
\end{proof}

\subsection{Graphs with no short holes and small forbidden subgraphs}
\label{sec:chordal}

Note that every hole $H$ of the graph $G$ defined in Theorem~\ref{thm:stcol-2} either 1. contains an edge $xy$ for vertices $x,y$ connected by a $k$-keeper, or 2.\ contains a path $x,v,-x$ for a vertex $v$ in a cluster $K$ connecting $\langle s, X, -X\rangle$ with $x \in X$ and $-x \in -X$.  A slight modification of $G$ can be used in the proof of Theorem~\ref{thm:stcol-2} so as to enlarge the hole $H$.  In case 1., $xy$ can be subdivided by inserting a vertex $z$ in such a way that $x,z$ and $z,y$ are connected by $k$-keepers.  Similarly, in case 2., dummy vertices not adjacent to any $p_h$ can be inserted into $X$ and $-X$ so as to increase the distance between $x$ and $-x$ in $H$.  Neither of these modifications generates a new hole in $G$.  Thus, in Theorem~\ref{thm:stcol-2} we can use a iterative modification of $G$ whose induced holes have length at least $h$.  The following corollary is then obtained.

\begin{corollary}
 For every $h \in O(1)$, \stcol{k} is \stp-complete when the input is restricted to $K_{k+2}$-free graphs whose induced holes have length at least $h$.
\end{corollary}

An interesting open question is what happens when $h$ grows to infinity, i.e., what is the complexity of star $k$-coloring a chordal graph or a chordal $K_{k+2}$-free graph.  In following sections we consider the star coloring and star choosability problems in some subclasses of chordal graphs, namely split, threshold, and block graphs.

Theorem~\ref{thm:stcol-2} also shows that the star $2$-coloring problem is hard for $\{K_4, C_4\}$-free graphs, which is a class of graphs defined by forbidding two small subgraphs.  Thus, another interesting question posed by Theorem~\ref{thm:stcol-2} is what happens when other small graphs are forbidden, so as to understand what structural properties can simplify the problem.  Following sections discuss the coloring problems from this perspective as well.  In particular, we study the problem for: every $H$-free graphs where $H$ has three vertices; a superclass of diamond-free graphs; and split and threshold graphs.  Before dealing with this restricted versions, we establish the complexity of the star $k$-choosability problem for $\{C_4, K_{k+2}\}$-free graphs.

\section{Complexity of the choosability problems}
\label{sec:choosability}
 
In this section we deal with the list version of the star and biclique-coloring problems.  The goal is to show that \stchose{k} and \bcchose{k} are \ptp-complete problems even when their inputs are restricted to $\{C_4, K_{k+2}\}$-free graphs.  Again, only one proof is required because the star and biclique-choosability problems coincide for $C_4$-free graphs.  In This opportunity, however, the proof is by induction on $k$.  That is, we first conclude that the \stchose{2} problem is \ptp-complete with a polynomial-time reduction from the \qsat{3}, and next we show that \stchose{k} can be reduced in polynomial time to the \stchose{(k+1)}.  The proof for $k = 2$ is similar to the proof of Theorem~\ref{thm:stcol-2}; however, we did not find an easy way to generalize it for $k > 2$ because long switchers generate graphs that are not star $k$-choosable.

\subsection{Keepers, switchers, clusters and forcers}

For the case $k = 2$ we require the keeper and cluster connections once again, and a new version of the long switcher.  We begin reviewing the main properties of keepers, switchers, and clusters with respect to the star choosability problem.

\begin{lemma}\label{lem:keeper choosability}
 Let $G$ be a graph, $L$ be a $k$-list assignment of $G$, and $K$ be a $k$-keeper connecting $v,w \in V(G)$ ($k \geq 2$).  Then, any $L$-coloring $\rho$ of $G \setminus (K \cup \{w\})$ can be extended into an $L$-coloring of $G$ in such a way that no monochromatic maximal star contains a vertex in $K$.
\end{lemma}

\begin{proof}
 Let $C_1, \ldots, C_{k-1}$ and $D = \{d_1, \ldots, d_{k-1}\}$ be the vertices of $K$ as in Definition~\ref{def:k-keeper connection}.  Extend $\rho$ into an $L$-coloring of $G$ such that $|\rho(D) \setminus \{\rho(v)\}| = k-1$, $|\rho(C_i)| = k$ for every \range{i}{1}{k-1}, and $\rho(w) \not\in \rho(D)$.  Since $L$ is a $k$-list assignment, such an extension can always be obtained.  Furthermore, no monochromatic maximal star has a vertex in $K$.
\end{proof}

\begin{lemma}\label{lem:switcher choosability}
 Let $G$ be a graph, $L$ be a $k$-list assignment of $G$, and $C$ be a $k$-switcher connecting $U \subset V(G)$ ($k \geq 2$).  Then, any $L$-coloring $\rho$ of $G \setminus C$ in which $|\rho(U)| \geq 2$ can be extended into an $L$-coloring of $G$ in such a way that no monochromatic maximal star has its center in $C$.
\end{lemma}

By the previous lemma, if $C$ is a $k$-switcher connecting $U$, then $G[U \cup C]$ is star $k$-choosable.  This property does not hold for long switchers because it is no longer true that a $k$-keeper connects vertices of the same color.  So, to avoid induced $C_4$'s, we need a new version of the long switcher.  This new switcher is defined only for $2$-colorings, and it is star $2$-choosable as desired.  We refer to this switcher as the \Definition{list switcher}.  In short, the difference between the long $2$-switcher and the list switcher is that the latter has no leafed vertices and $2$-keepers are replaced by edges (see Figure~\ref{fig:switcher-forcer}~(a)).  Its definition is as follows.

\begin{defn}[list switcher]\label{def:list switcher}
  Let $G$ be a graph, and $W = \{w_1, \ldots, w_h\}$ be a set of vertices of $G$ with $h \geq 2$.  Say that $S \subset V(G)$ is a \Definition{list switcher connecting $W$} when $S$ can be partitioned into an independent set $U = \{u_1, \ldots, u_h\}$ and a $2$-switcher $C$ in such a way that $C$ connects $U$, $w_iu_i \in E(G)$ for \range{i}{1}{h}, and there are no more edges adjacent to vertices in $U$.
\end{defn}

The following lemma is equivalent to Lemma~\ref{lem:long switcher properties} for list switchers.

\begin{lemma}\label{lem:list switcher properties}
 Let $G$ be a graph and $S$ be a list switcher connecting $W \subset V(G)$.  Then,
\begin{enumerate}[(i)]
  \item no induced $C_4$ or $K_4$ of $G$ contains a vertex of $S$,
  \item $|\rho(W)| \geq 2$ for any star $k$-coloring $\rho$ of $G$, and \label{lem:list switcher properties:colors}
  \item for any $2$-list assignment $L$ of $G$, every $L$-coloring $\rho$ of $G \setminus S$ in which $|\rho(W)| \geq 2$ can be extended into an $L$-coloring of $G$ in such a way that no monochromatic maximal star contains a vertex in $S$.\label{lem:list switcher properties:extension}
\end{enumerate}
\end{lemma}

\begin{proof}
 We only prove (\ref{lem:list switcher properties:extension}).  Let $w_1, \ldots, w_h$, and $u_1, \ldots, u_h$ be as in Definition~\ref{def:list switcher}, and $\{x, y\}$ be the $2$-switcher connecting $\{u_1, \ldots, u_h\}$.  Suppose, without loss of generality, that $\rho(w_1) \neq \rho(w_2)$, and observe that either $\rho(w_1) \not\in L(x)$ or $\rho(w_2) \not\in L(y)$ or $L(x) = L(y) = \{\rho(w_1),\rho(w_2)\}$.  In this setting, extend $\rho$ to include $x$ and $y$ in such a way that $\rho(x) \neq \rho(y)$, $\rho(x) \neq \rho(w_1)$, and $\rho(y) \neq \rho(w_2)$.  Following, extend $\rho$ into an $L$-coloring of $G$ such that $\rho(u_1) \neq \rho(y)$, $\rho(u_2) \neq \rho(x)$, and $\rho(u_i) \neq \rho(w_i)$ for \range{i}{3}{h}.  It is not hard to see that no monochromatic maximal star contains a vertex in $S$.
\end{proof}

Finally, the proof of statement (\ref{lem:cluster properties:extension}) of Lemma~\ref{lem:cluster properties} implies the following lemma.

\begin{lemma}\label{lem:cluster choosability}
  Let $G$ be a graph, $L$ be a $2$-list assignment of $G$, and $K$ be an $\ell$-cluster connecting $\langle s, X, -X\rangle$ for $\{s\} \cup X \cup -X \subseteq V(G)$.  Then, any $L$-coloring $\rho$ of $G \setminus K$ can be extended into an $L$-coloring of $G$ in such a way that $\rho(s) \not\in \rho(K)$, and no monochromatic maximal star of $G$ has its center in $K$.
\end{lemma}

\begin{figure}
 \centering
 \begin{tabular}{c@{\hspace{1cm}}c}
  \includegraphics{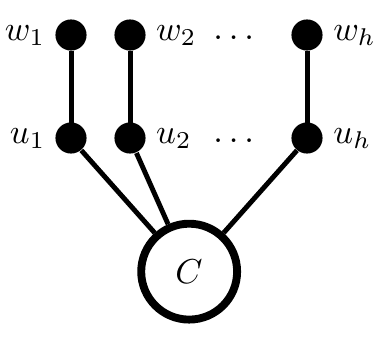} & \includegraphics{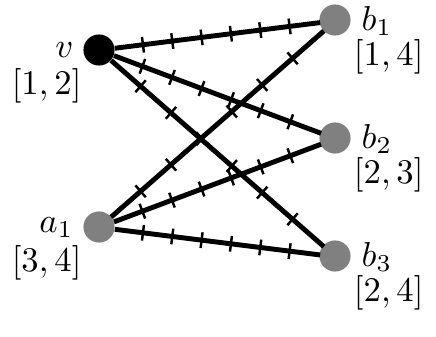} \\
  (a) & (b) 
 \end{tabular}
 \caption{(a) A list switcher connecting $\{u_1, \ldots, u_h\}$. (b) A $2$-forcer connecting $v$ with a $2$-list assignment $L$; note that $1$ is the unique $L$-admissible color for $v$.}\label{fig:switcher-forcer}
\end{figure}

Besides the $2$-keepers, list switchers, and clusters, we use forth kind of connection that can be used to force the color of a given vertex when an appropriate list assignment is chosen.  This connection is called the \emph{forcer} and, contrary to the other connections, it connects only one vertex.

\begin{defn}[$k$-forcer]\label{def:k-forcer connection}
  Let $G$ be a graph and $v \in V(G)$.  Say that $F \subset V(G)$ is a \emph{$k$-forcer connecting $v$} ($k \geq 2$) when $F$ can be partitioned into sets of leafed vertices $A$ and $B$, and cliques $C(a,b)$ for $a \in A \cup \{v\}$ and $b \in B$ in such a way that $|A| = k-1$, $|B|=k^{k}-1$, $C(a,b)$ is a $k$-switcher connecting $\{a,b\}$, and there are no more edges incident to vertices in $A \cup B$.
\end{defn}

Let $L$ be a $k$-list assignment of $G$ and $F$ be a $k$-forcer connecting $v \in V(G)$.  We say that $c \in L(v)$ is \emph{$L$-admissible for $v$} when there is an $L$-coloring $\rho$ of $G$ in which $\rho(v) = c$ and no monochromatic maximal star has its center in $F$.  Clearly, if $c$ is $L$-admissible for $v$, then any $L$-coloring $\rho$ of $G \setminus F$ in which $\rho(v) = c$ can be extended into an $L$-coloring of $G$ in such a way that no monochromatic maximal star has its center in $F$.  The main properties of forcers are summed up in the following lemma (see Figure~\ref{fig:switcher-forcer}~(b)).

\begin{lemma}\label{lem:forcer properties}
  Let $G$ be a graph and $F$ be a $k$-forcer connecting $v \in V(G)$ ($k \geq 2$).  Then, 
  \begin{enumerate}[(i)]
    \item no induced $K_{k+2}$ or $C_4$ of $G$ contains a vertex in $F$,\label{lem:forcer properties:forbidden}
    \item for every $k$-list assignment $L$ of $G$ there is an $L$-admissible color for $v$, and\label{lem:forcer properties:extension}
    \item every $k$-list assignment $L$ of $G \setminus F$ can be extended into a $k$-list assignment of $G$ in which $v$ has a unique $L$-admissible color.\label{lem:forcer properties:color}
  \end{enumerate}
\end{lemma}

\begin{proof}
 Let $A$, $B$, and $C(a,b)$ be as in Definition~\ref{def:k-forcer connection}, and define $A^* = A \cup \{v\}$.  Statement (\ref{lem:forcer properties:forbidden}) follows from the fact that no pair of vertices in $A^* \cup B$ have a common neighbor.
 
 (\ref{lem:forcer properties:extension}) Let $H$ be the complete bipartite graph with bipartition $\{h_a \mid a \in A^*\}\{h_b \mid b \in B\}$, and $M$ be a $k$-list assignment of $H$ where $M(h_a) = L(a)$ for every $a \in A^*\cup B$.  In~\cite{MarxTCS2011} it is proven that $H$ admits a vertex $M$-coloring $\gamma$.  Define $\rho$ as any $L$-coloring of $G$ in which $\rho(a) = \gamma(h_a)$ for $a \in A^* \cup B$, where $k$-switchers and leaves are colored according to Lemma~\ref{lem:switcher choosability} and Observation~\ref{obs:leafed vertex}, respectively.  The coloring of the $k$-switchers is possible because $\gamma$ is a vertex $M$-coloring.  By Observation~\ref{obs:leafed vertex}, no vertex in $A \cup B$ is the center of a maximal star, while by Lemma~\ref{lem:switcher choosability}, no vertex in $C(a,b)$ is the center of a maximal monochromatic star for $a \in A^*$ and $b \in B$.  That is, $\rho(v)$ is $L$-admissible for $v$.
 
 (\ref{lem:forcer properties:color})  Extend $L$ into a $k$-list assignment of $G$ such that 1.\ $L(a) \cap L(a') = \emptyset$ for every pair of vertices $a, a' \in A^*$, 2.\ $\mathcal{L}(B) = \{L(b) \mid b \in B\}$ is a family of different subsets included in $L(A^*)$ such that $|L(b) \cap L(a)| = 1$ for every $a \in A^*$ and $b \in B$, and 3.\ $L(C(a,b)) = L(b)$ for every $a \in A^*$, and $b \in B$.  Define $H$ and $M$ as in statement (\ref{lem:forcer properties:extension}).  By statement~(\ref{lem:forcer properties:extension}), there is an $L$-coloring $\rho$ of $G$ that contains no monochromatic maximal star with center in $F$.  Since $C(a,b)$ is a block of $G$, it follows that $\rho(C(a,b)) = L(b)$ by Observation~\ref{obs:block coloring}.  Then, since no maximal star with center in $C(a,b)$ is monochromatic, it follows that $\rho(a) \neq \rho(b)$ for every $a \in A^*$, $b \in B$.  Thus, if $\gamma$ is the coloring such that $\gamma(h_a) = \rho(a)$ for every $a \in A^*\cup B$, then $\gamma$ is a vertex $M$-coloring of $H$.  Consequently, as proven in~\cite{MarxTCS2011}, $\gamma(h_v) = \rho(v)$ is the unique color of $L(v)$ that belongs to the subset of $L(A^*) \not\in \mathcal{L}(B)$.
 \end{proof}

\subsection{Hardness of the star choosability problem}

A problem $P$ is \ptp when the problem of authenticating a negative certificate of $P$ is \stp~\cite{Papadimitriou1994}.  For the star $k$-choosability problem, a $k$-list assignment of $G$ can be taken as the negative certificate.  Using arguments similar to those in Section~\ref{sec:general case:2-stp} for star $k$-colorings, it is not hard to see that it is a \stp problem to authenticate whether a graph $G$ admits no $L$-colorings for a given $k$-list assignment $L$.  Therefore, \stchose{k} is \ptp.  In this section we establish the hardness of \stchose{k}.  For $k=2$ we reduce the complement of an instance of \qsat{3} into an instance of \stchose{2}.  Then, we proceed by induction showing how to reduce an instance of \stchose{k} into an instance of \stchose{(k+1)} for every $k \geq 2$.  

The proof for the case $k=2$ is, in some sense, an extension of Theorem~\ref{thm:stcol-2}.  The goal is to force the true literals of $z$ variables to have the same color as $s$, so that a monochromatic maximal star centered at $s$ appears when the formula is false.  

\begin{theorem}\label{thm:stchose-2}
 \stchose{2} is \ptp-hard, and it remains \ptp-hard even when its input is restricted to $\{C_4, K_4\}$-free graphs.
\end{theorem}

\begin{proof}
  The hardness of \stchose{2} is obtained by reducing the complement of \qsat{3}.  That is, given a $3$-DNF formula $\phi(\VEC z, \VEC x, \VEC y)$ with $\ell$ clauses $\VP P_1, \ldots, \VP P_\ell$, and $n+m+o$ variables $\VEC x= \VP x_1, \ldots, \VP x_n$, $\VEC y = \VP y_1, \ldots, \VP y_m$, $\VEC z = \VP z_1, \ldots, \VP z_o$, we build a graph $G$ that is $2$-list-choosable if and only if $(\forall\VEC z)(\exists \VEC x)(\forall \VEC y)\phi(\VEC z, \VEC x, \VEC y)$ is true.  For the sake of simplicity, in this proof we use $i$, $j$, $h$, and $f$ to refer to values in $\{1, \ldots, n\}$, $\{1, \ldots, m\}$, $\{1, \ldots, \ell\}$, and $\{1, \ldots, o\}$, respectively.  
  
  Graph $G$ is similar to the graph in Theorem~\ref{thm:stcol-2}.  Its vertex set is again divided into connection, inner, and leaf vertices. In turn, \emph{connection} vertices are divided into a set $P = \{p_1, \ldots, p_\ell\}$, a set $X = \{x_1, \ldots, x_n\}$, sets $X_i$, $-X_i$, $Y_j$, $-Y_j$, $Z_f$, and $-Z_f$ with $\ell$ vertices each, and two vertices $s,t$.  Let $L_X = \bigcup_i (X_i \cup (-X_i))$, $L_Y = \bigcup_j (Y_j \cup (-Y_j))$, and $L_Z = \bigcup_f (Z_f \cup (-Z_f))$.
  
  \emph{Inner} vertices are those vertices included in $2$-keepers, list switcher, clusters and $2$-forcers.  The following rules define inner vertices and the edges between connection vertices.  \textbf{Edges:} $s$ is adjacent to all the vertices in $P$, and if $\VP x_i$ (resp.\ $\overline{\VP x_i}$, $\VP y_j$, $\overline{\VP y_j}$, $\VP z_f$, $\overline{\VP z_f}$) is a literal of $\VP P_h$, then $p_h$ is adjacent to $-X_i[h]$ (resp.\ $X_i[h]$, $-Y_j[h]$, $Y_j[h]$, $Z_f[h]$, $-Z_f[h]$).  \textbf{Keepers:} $s$ and $t$ are connected by a $2$-keeper.  \textbf{List switchers:}  there are list switchers connecting \{$x_i$, $-X_i[h]$\} and \{$X_i[h]$, $-X_i[h]$\} (for every $i, h$), $\{Z_f[h], -Z_f[h]\}$ (for every $f, h$), and $\{p_h, t\}$ (for every $h$).  \textbf{Clusters:}  there are $\ell$-clusters connecting $\langle s, X_i, -X_i\rangle$, $\langle s, Y_j, -Y_j\rangle$, and $\langle s, Z_f, -Z_f\rangle$. \textbf{Forcers:} there are $2$-forcers connecting each vertex of $Z_f$ (for every $f$) and each vertex of $L_Y$.

  Finally, every connection vertex other than $s$ is leafed.  This ends up the construction of $G$ (see Figure~\ref{fig:choosing}), which can be easily computed from $\phi(\VEC z, \VEC x, \VEC y)$ in polynomial time.   Arguments similar to those in Theorem~\ref{thm:stcol-2} are enough to conclude that $G$ is $\{C_4, K_4\}$-free.  

\begin{figure}
  \centering
  \includegraphics{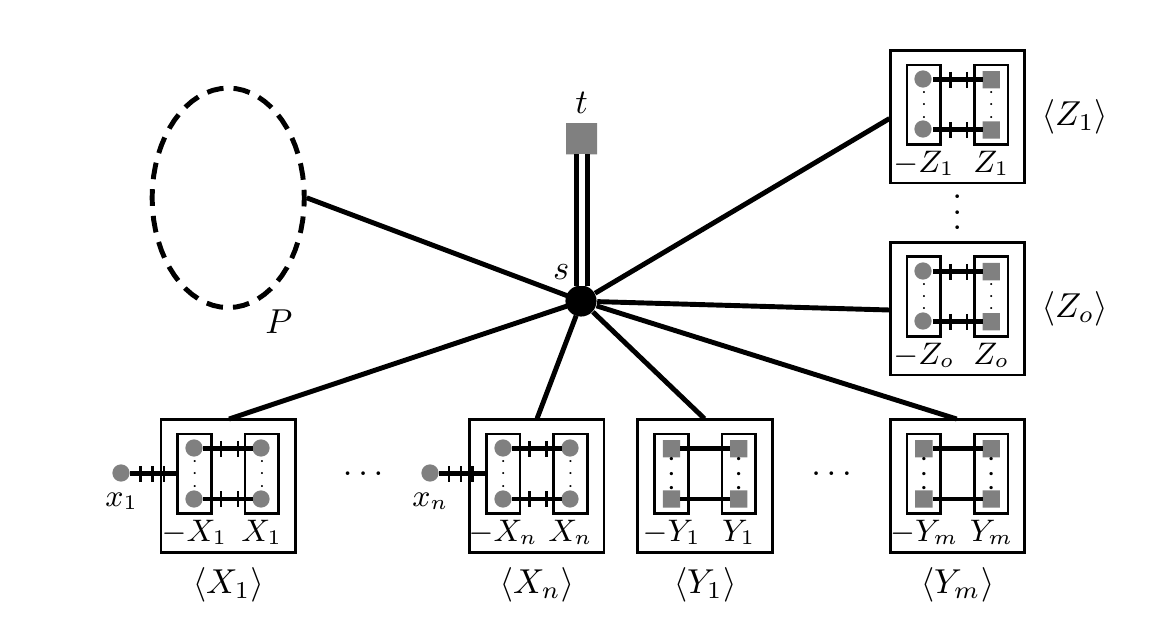}
  \caption{Schema of the graph obtained from $\phi$ in Theorem~\ref{thm:stchose-2}; for the sake of simplicity, we omit the edges from $P$ to $L_X \cup L_Y \cup L_Z$.  Square vertices represent vertices connected to a $2$-forcer.}\label{fig:choosing}
\end{figure}

  We now show that $(\forall\VEC z)(\exists\VEC x)(\forall\VEC y)\phi(\VEC z, \VEC x, \VEC y)$ is true if and only if $G$ is $2$-choosable.  We first show that if $(\forall\VEC z)(\exists\VEC x)(\forall\VEC y)\phi(\VEC z, \VEC x, \VEC y)$ is true, then $G$ admits a star $L$-coloring $\rho$ for any $2$-list assignment $L$.  The $L$-coloring $\rho$ is obtained by executing the following algorithm.
  \begin{description}
    \item[Step 1:] For every $w$ connected to a forcer, let $\rho(w)$ be $L$-admissible for $w$.  Such a color always exists by statement~(\ref{lem:forcer properties:extension}) of Lemma~\ref{lem:forcer properties}.  Suppose, w.l.o.g., that $\rho(t) = 1$ and let $\nu(\VEC z)$ be a valuation of $\VEC z$ such that $\nu(\VP z_f) = 1$ if and only if $\rho(Z_f) = \{1\}$.  
    
    \item[Step 2:] By hypothesis, $\nu$ can be extended to include $\VEC x$ so that $(\forall \VP y)\nu(\phi(\VEC z, \VEC x, \VEC y))$ is true.  If $L(x_i) \neq L(-X_i[h])$ or $L(X_i[h]) \neq L(-X_i[h])$ for some $h$, then:
    \begin{description}
      \item [Step 2.1:] Let $\rho(X_i[h]) \neq \rho(-X_i[h])$ in such a way that $\rho(X_i[h]) = 1$ if and only if $1 \in L(X_i[h])$ and $\nu(\VP x_i) = 1$, while $\rho(-X_i[h]) = 1$ if and only if $1 \in L(-X_i[h])$ and $\nu(\VP x_i) = 0$.
      \item [Step 2.2:] Let $\rho(x_i) \neq \rho(-X_i[h])$.
      \item [Step 2.3:] Let $\rho(-X_i[k]) \neq \rho(x_i)$ and $\rho(X_i[k]) \neq \rho(-X_i[k])$ for every $k \neq h$.
    \end{description}
    If $L(x_i) = L(-X_i) = L(X_i)$, then:
    \begin{description}
      \item [Step 2.4:] Let $\rho(x_i) = \rho(X_i) = 1$ if and only if $1 \in L(x_i)$ and $\nu(\VP x_i) = 1$, and $\rho(-X_i) \neq \rho(x_i)$.
    \end{description}
    Note that, whichever case gets executed, $\rho(x_i) \not\in \rho(-X_i)$ and $\rho(-X_i[h]) \neq \rho(X_i[h])$ for every $h$.
    
    \item[Step 3:] Let $\rho(-Z_f[h]) \neq \rho(Z_f[h])$ and $\rho(p_h) \neq 1$.  
    
    \item[Step 4:] Let $\rho$ for $s$, leaves, and inner vertices be as in Observation~\ref{obs:leafed vertex}, and Lemmas~\ref{lem:keeper choosability}, \ref{lem:list switcher properties}, \ref{lem:cluster choosability}~and~\ref{lem:forcer properties}.  Observe that this is always possible.  In particular, observe that every pair of vertices connected by a list switcher have different colors, while every vertex connected to a forcer has an $L$-admissible color.
  \end{description}

  We claim that $\rho$ is a star $L$-coloring of $G$.  Let $\{w\}S$ be any maximal star of $G$.  By Observation~\ref{obs:leafed vertex} and Lemmas~\ref{lem:keeper choosability}, \ref{lem:list switcher properties},  \ref{lem:cluster choosability}~and~\ref{lem:forcer properties}, $\{w\}S$ is not monochromatic when $w \neq s$.  Suppose, for the rest of the proof, that $w = s$ and $\{s\}(S \setminus P)$ is monochromatic.  By Lemma~\ref{lem:keeper choosability}, this implies that $t \in S$, thus $\rho(s) = \rho(t) = 1$.  Also, by Lemma~\ref{lem:cluster choosability}, $S$ intersects no $k$-cluster, thus $S \subseteq \{t\} \cup L_X \cup L_Y \cup L_Z \cup P$.  Moreover, by statement~(\ref{lem:cluster properties:stars}) of Lemma~\ref{lem:cluster properties}, $S \cap (X_i \cup -X_i)$ equals either $X_i$ or $-X_i$, $S \cap (Y_j \cup -Y_j)$ equals either $Y_j$ or $-Y_j$, and $S \cap (Z_f \cup -Z_f)$ is either $Z_f$ or $-Z_f$, for every $i$, $j$, and $f$.  Extend $\nu$ to $\VEC y$ so that $\nu(\VP y_j) = 1$ if and only if $Y_j \subset S$.  By hypothesis, $\nu(\phi(\VEC z, \VEC x, \VEC y)) = 1$, thus there is some clause $\VP P_h$ whose literals are all true according to $\nu$.  If $p_h$ has some neighbor in $-Y_j$, then $\nu(\VP y_j) = 1$, thus $Y_j \subset S$ and $-Y_j \cap S = \emptyset$.  If $p_h$ has some neighbor $-z_h \in -Z_f$, then $\nu(\VP z_f) = 1$ which means, by the way $\nu$ is defined for $\VEC z$ in Step~1, that $\rho(Z_f) = \{1\}$.  Consequently, by Step~3, $\rho(-z_h) \neq 1$, i.e., $-z_h \not\in S$.  Similarly, if $p_h$ has some neighbor in $Z_f$, then $\nu(\VP z_f) = 0$ which means that $\rho(Z_f) \neq \{1\}$.  Thus, there must exist at least one vertex $z_f \in Z_f$ with $\rho(z_f) \neq 1$.  Then, since $\rho(S) = 1$, it follows that $Z_f \not\subset S$.   Finally, if $p_h$ has a neighbor in $-X_i$, then $\nu(\VP x_i) = 1$, thus $\rho(-X_i[h]) \neq 1$ for some $h$ by either Step~2.1 or Step~2.4.  Hence, $\rho(-X_i) \neq \{1\}$, thus $-X_i \not\subset S$.  Analogously, $p_h$ has no neighbors in $X_i \cap S$.  Summing up, since $P \cup \{t\}$ is an independent set, it follows that $p_h$ has no neighbors in $S$, thus $p_h \in S$ and $\{s\}S$ is not monochromatic by Step~3.

  For the converse, suppose $G$ is star $2$-choosable, and consider any valuation $\nu$ of $\VEC z$.  Define $L$ to be a $2$-list assignment of $G$ so that $\nu(\VP z_f)$ is the unique color admissible for all the vertices in $Z_f$, $1$ is the unique color admissible for all the vertices in $L_Y \cup \{t\}$, and $L(w) = \{0,1\}$ for every vertex not connected to a $2$-forcer.  By statement~(\ref{lem:forcer properties:color}) of Lemma~\ref{lem:forcer properties}, such list assignment $L$ always exists.  Let $\rho$ be a star $L$-coloring of $G$ and extend $\nu$ to include $\VEC x + \VEC y$ in its domain so that $\nu(\VP x_i) = \rho(v_i)$.  Note that $\nu(\VP y_j)$ can take any value from $\{0,1\}$, so it is enough to prove that $\nu(\phi(\VEC x, \VEC y, \VEC z)) = 1$.  Define $V_X = \bigcup_i((X_i \mid \nu(\VP x_i) = 1) \cup (-X_i \mid \nu(\VP x_i) = 0))$, $V_Y = \bigcup_j((Y_j \mid \nu(\VP y_j) = 1) \cup (-Y_j \mid \nu(\VP y_j) = 0))$, and $V_Z = \bigcup_f((Z_f \mid \nu(\VP z_f) = 1) \cup (-Z_f \mid \nu(\VP z_f) = 0))$, and let $S = \{s\}(\{t\} \cup V_X \cup V_Y \cup V_Z)$.  As in Theorem~\ref{thm:stcol-2}, it can be observed that (i) $S$ is a monochromatic star and (ii) every vertex in $N(s) \setminus P$ is either adjacent or equal to a vertex in $S$.  Thus, since $\rho$ is a star $L$-coloring of $G$, there must be some vertex $p_h$ adjacent to no vertex in $V_X \cup V_Y \cup V_Z$.  Moreover, such vertex $p_h$ corresponds to some clause $\VP P_h$ whose literals are all true by the way $\nu$ is defined.
\end{proof}

The proof for $k > 2$ is by induction, i.e., we reduce \stchose{k} into \stchose{(k+1)} for every $k \geq 2$.  Roughly speaking, the idea of the reduction is to insert a vertex $z$ that forbids every vertex of the reduced graph to have the same color as $z$.

\begin{theorem} \label{thm:stchose-k}
  \stchose{k} is \ptp-complete for every $k \geq 2$, and it remains \ptp-complete when the input is restricted to $\{C_4, K_{k+2}\}$-free graphs.
\end{theorem}

\begin{proof}
  The proof is by induction on $k$.  The base case $k = 2$ corresponds to Theorem~\ref{thm:stchose-2}.  For the inductive step, we show how to transform a $\{C_4, K_{k+2}\}$-free graph $G_k$ into a $\{C_4, K_{k+3}\}$-free graph $G_{k+1}$ so that $G_k$ is star $k$-choosable if and only if $G_{k+1}$ is star $(k+1)$-choosable.  
  
  The vertices of $G_{k+1}$ are divided into connection and inner vertices.  Connection vertices comprise a set $W$ inducing $G_{k}$ and a vertex $z$.  Inner vertices are included in $(k+1)$-forcers or $(k+1)$-switchers connecting connection vertices.  There is a $(k+1)$-forcer connecting $z$, and a $(k+1)$-switcher connecting $\{z,w\}$ for every $w \in W$.  Let $C(w)$ be the $(k+1)$-switcher connecting $\{w,z\}$, i.e., $C(w)\cup\{w\}$ and $C(w) \cup \{z\}$ are cliques of $G_{k+1}$.  By statement (\ref{lem:forcer properties:forbidden}) of Lemmas \ref{lem:switcher properties}~and~\ref{lem:forcer properties}, $G_{k+1}$ is $\{K_{k+3}, C_4\}$-free.  

  Suppose $G_k$ is star $k$-choosable.  Let $L_{k+1}$ be a $(k+1)$-list assignment of $G_{k+1}$, and $c(z) \in L(z)$ be $L$-admissible for $z$.  Recall that $c(z)$ always exists by statement~(\ref{lem:forcer properties:extension}) of Lemma~\ref{lem:forcer properties}.  Define $L_k$ as a $k$-list assignment of $G_{k+1}[W]$ such that $L_k(w) \subseteq L_{k+1}(w) \setminus \{c(z)\}$ for $w \in W$.  By hypothesis, there is a star $L_k$-coloring $\rho$ of $G_{k+1}[W]$.  Define $\sigma$ to be the $L_{k+1}$-coloring of $G_{k+1}$ such that $\sigma(w)=\rho(w)$ for $w \in W$ and $\sigma(z) = c(z)$.  Inner vertices are colored according to Lemma~\ref{lem:switcher choosability} and statement~(\ref{lem:forcer properties:extension}) of Lemma~\ref{lem:forcer properties}.  Clearly, if $\{w\}S$ is a maximal star of $G_{k+1}$ and $w$ is a connection vertex, then either $w = z$ or $\{w\}S$ includes a maximal star of $G_{k}[W]$.  Whichever the case, $\{w\}S$ is not monochromatic, i.e., $\sigma$ is a star coloring of $G_{k+1}$.

  For the converse, let $L_k$ be an $k$-list assignment of $G_{k+1}[W]$ and take a color $c\not\in L(V(G_{k+1}))$.  Define $L_{k+1}$ as any $(k+1)$-list assignment of $G_{k+1}$ such that $c$ is the unique $L_{k+1}$-admissible color for $z$, and $L_{k+1}(w) = L_{k+1}(C(w)) = L_{k}(w) \cup \{c\}$ for every $w \in W$.  Such list assignment always exists by statement~(\ref{lem:forcer properties:color}) of Lemma~\ref{lem:forcer properties}.  Let $\sigma$ be a star $L_{k+1}$-coloring of $G_{k+1}$.  By construction, $\sigma(z) = \{c\}$, and by Lemma~\ref{lem:switcher properties}, $c \not\in \sigma(W)$.  Hence, the restriction $\rho$ of $\sigma$ to $W$ is an $L_k$-coloring of $G_{k+1}[W]$.  Moreover, if $\{w\}S$ is a maximal star of $G_{k+1}[W]$, then $\{w\}(S \cup \{x\})$ is a maximal star of $G_{k+1}$, for every $x \in C(w)$.  Since $C(w)$ is a block of $G_{k+1}$ and $L(C(w)) = L(w)$, it follows that $\rho(w) = \sigma(w) \in \sigma(C(w))$.  Hence, $\{w\}S$ is not monochromatic.
\end{proof}

\section{Forbidding graphs of order 3} 
\label{sec:small forbiddens}

The previous sections dealt with time complexity of the star and biclique coloring and choosability problems.  The remaining of the article is devoted to these problems in other restricted classes of graphs.  As discussed in Section~\ref{sec:chordal}, we are interested in classes of graphs that are related to chordal graphs or can be defined by forbidding small induced subgraphs.  In this section, we study the classes of $H$-free graphs, for every graph $H$ on three vertices.

There are four graphs with exactly three vertices, namely $K_3$, $P_3$, $\overline{P_3}$, and $\overline{K_3}$.  The following theorem shows that $K_3$-free graphs are star $2$-choosable.

\begin{theorem}\label{thm:stcol triangle-free}
  Every $K_3$-free graph is star $2$-choosable.  Furthermore, for any $2$-list assignment, a star $L$-coloring can be obtained in linear time.
\end{theorem}

\begin{proof}
  Let $L$ be a $2$-list assignment of a $K_3$-free graph $G$, $T$ be a rooted tree subgraph of $G$ with $V(T) = V(G)$, $r$ be the root of $T$, and $p(v)$ be the parent of $v$ in $T$ for each $v \in V(G) \setminus \{r\}$.  Define $\rho$ to be an $L$-coloring of $G$ where $\rho(r) \in L(r)$ and $\rho(v) \in L(v) \setminus \{\rho(p(v))\}$ for every $v \in V(G) \setminus \{r\}$.  Since $G$ is $K_3$-free, $\{v\}S$ is a maximal star of $G$ for $v \in V(G)$ only if $S = N(v)$, hence $\{v\}S$ is not monochromatic.  Observe that a BFS traversal of $G$ is enough to compute $\rho$, thus $\rho$ is computed in linear time from $G$.
\end{proof}

As a corollary, we obtain that $\{C_4, K_3\}$-free graphs are biclique $2$-choosable also.  However, this corollary can be easily strengthened so as to include those $K_3$-free graphs that are \emph{biclique-dominated}.  A graph $G$ is \Definition{biclique-dominated} when every maximal biclique is either a star or has a false dominated vertex.  Some interesting classes of graphs are $K_3$-free and biclique-dominated, including hereditary biclique-Helly graphs~\cite{EguiaSoulignacDMTCS2012}.

\begin{theorem}
  Every $K_3$-free graph that is biclique-dominated is biclique $2$-choosable.  Furthermore, for any $2$-list assignment, a biclique $L$-coloring can be computed in polynomial time.
\end{theorem}

\begin{proof}
  Let $L$ be a $2$-list assignment of a $K_3$-free graph $G$ that is biclique-dominated.  The algorithm for biclique $L$-coloring $G$ has two steps.  First, apply Theorem~\ref{thm:stcol triangle-free} on $G$ so as to obtain a star $L$-coloring $\rho$ of $G$.  Second, traverse each vertex $w$ and, for each $v$ that is false dominated by $w$, change $\rho(v)$ with any color in $L(v) \setminus \rho(w)$.  (It is not important if $\rho(v)$ or $\rho(w)$ are later changed when other vertices are examined.) The coloring thus generated is a biclique $L$-coloring.  Indeed, if a maximal biclique contains a false dominated vertex $v$, then it also contains the vertex $w$ such that $\rho(v)$ was last changed in the second step while traversing $w$.  Since false domination is a transitive relation, it follows that $\rho(v) \neq \rho(w)$ when the second step is completed.  On the other hand, if $S$ is a maximal biclique with no false dominated vertices, then $S$ is a star.  Since the colors of the vertices of $S$ are not affected by the second step, we obtain that $S$ is not monochromatic.  It is not hard to see that the algorithm requires polynomial time.
\end{proof}

Coloring a connected $P_3$-free graph is trivial because the unique connected $P_3$-free graph $G$ with $n$ vertices is $K_n$.  Thus, $ch_B(G) = \chi_B(G) = \chi_S(G) = ch_S(G) = n$.

\begin{theorem}
 If $G$ is a connected $P_3$-free graph with $n$ vertices, then $ch_B(G) = \chi_B(G) = \chi_S(G) = ch_S(G) = n$.
\end{theorem}

The case of $\overline{P_3}$-free graphs, examined in the next theorem, is not much harder.

\begin{theorem}
 If $G$ is a $\overline{P_3}$-free graph with $k$ universal vertices, then $ch_B(G) = \chi_B(G) = \chi_S(G) = ch_S(G) = \max\{2,k\}$.
\end{theorem}

\begin{proof}
 Let $K$ be the set of universal vertices of $G$.  Clearly, $K$ is a block of $G$, thus $ch_B(G) \geq k$ and $ch_S(G) \geq k$ by Observation~\ref{obs:block coloring}.  For the other bound, let $L$ be a $k$-list assignment of $G$, and $B_1, \ldots, B_n$ be the sets of vertices that induce components of $\overline{G} \setminus K$.  Define $\rho$ as an $L$-coloring of $G$ such that $|\rho(K)| = k$ and $|\rho(B_i)| = 2$ for \range{i}{1}{n}.  Note that $B_i$ is a set of false twin vertices (\range{i}{1}{j}) because $B_i$ is a clique of $\overline{G}$.  Thus, every maximal star or biclique $S$ is formed by two vertices of $K$ or it contains a set $B_i$ for some \range{i}{1}{n}.  Whichever the case, $S$ is not monochromatic, thus $\rho$ is a star and biclique $L$-coloring.
\end{proof}

The remaining class is the class of $\overline{K_3}$-free graphs.  By definition, if $G$ is $\overline{K_3}$-free, then every maximal star and every maximal biclique of $G$ has $O(1)$ vertices.  Thus, it takes polynomial time to determine if an $L$-coloring of $G$ is a star or biclique coloring, for any $k$-list assignment $L$.  Hence, when restricted to $\overline{K_3}$-free graphs, the star and biclique $k$-coloring problems belong to \NP, while the star and biclique $k$-choosability problems belong to \ptwop.  The next theorem shows that, when $k \geq 3$, the choosability problems are \ptwop-complete even when the input is further restricted to co-bipartite graphs.

\begin{theorem}\label{thm:co-bipartite choosability}
 \stchose{k} and \bcchose{k} are \ptwop-complete for every $k \geq 3$ when the input is restricted to co-bipartite graphs.
\end{theorem}

\begin{proof}
 The proof is obtained by reducing the problem of determining if a connected bipartite graph with no false twins is vertex $k$-choosable, which is known to be \ptwop-complete~\cite{GutnerTarsiDM2009}.  Let $G$ be a connected bipartite graph with no false twins, $XY$ be a bipartition of $G$, and $k \in \mathbb{N}$.  Define $H$ to be the bipartite graph obtained from $G$ by inserting, for every $vw \in E(G)$, the stars $\{a_i(vw)\}A_i(vw)$ (\range{i}{1}{4}) with $|A_i(vw)| = k-1$ and the edges $va_1(vw)$, $va_3(vw)$, $wa_2(vw)$, $wa_4(vw)$, $a_1(vw)a_2(vw)$, and $a_3(vw)a_4(vw)$ (see Figure~\ref{fig:co-bip choose}).  We claim that $G$ is vertex $k$-choosable if and only if $\overline{H}$ is star (resp.\ biclique) $k$-choosable.  
 
\begin{figure}
  \centering\includegraphics{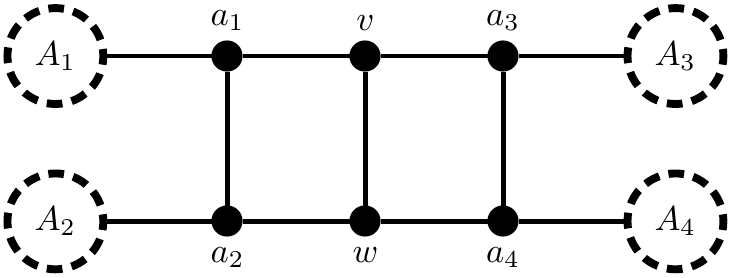}\caption{Transformation applied to $vw$ in Theorem~\ref{thm:co-bipartite choosability}; each independent set has $k-1$ vertices.}\label{fig:co-bip choose}
\end{figure}

 Suppose first that $G$ is vertex $k$-choosable, and let $L$ be a $k$-list assignment of $H$ and $M$ be the restriction of $L$ to $V(G)$.  By hypothesis, $G$ admits a vertex $M$-coloring $\rho$.  Define $\sigma$ to be any vertex $L$-coloring of $H$ so that $\sigma(v) = \rho(v)$ for $v \in V(G)$, and $|\sigma(A_i(vw)\cup\{a_i(vw)\})| = k$ for every $vw \in E(G)$ and every \range{i}{1}{4}.  It is not hard to see that such a coloring always exists.  Clearly, every maximal star (resp.\ biclique) $S$ of $\overline{H}$ is formed by two twins of $\overline{H}$ or it contains two vertices that are adjacent in $H$.  In the latter case $S$ is not $\sigma$-monochromatic because $\sigma$ is a vertex coloring of $H$, while in the former case $S$ is not $\sigma$-monochromatic because both of its vertices must belong to $A_i(vw)$, as $G$ has no false twins, for some $vw \in E(G)$ and some \range{i}{1}{4}.  
  
 For the converse, suppose $\overline{H}$ is star (resp.\ biclique) $k$-choosable, and let $M$ be a $k$-list assignment of $G$.  Define $\sigma$ to be a star (resp.\ biclique) $L$-coloring of $\overline{H}$, for the $k$-list assignment $L$ of $H$ where $L(a) = L(v) = M(v)$ for every $vw \in E(G)$ with $v \in X$, and every $a \in A_i(vw) \cup \{a_i(vw)\}$ with \range{i}{1}{4}.  Suppose, to obtain a contradiction, that $\sigma(v) = \sigma(w)$ for some $vw \in E(G)$ with $v \in X$ and $w \in Y$.  Then, for every $a \in A_i(vw)$ (\range{i}{1}{4}), we obtain that $\sigma(a) \neq \sigma(v)$ because $\{a\}\{v, w\}$ is a maximal star (resp.\ biclique) of $\overline{H}$.  Hence, since $A_i(vw)$ is a block of $\overline{H}$, we obtain by Observation~\ref{obs:block coloring} that $\sigma(A_i(vw)) = L(v) \setminus \{\sigma(v)\}$ for every $\range{i}{1}{4}$.  Consequently, since $\{b\}\{a, a_1(vw)\}$ is a maximal star (resp.\ biclique) for every $b \in A_2(vw)$ and every $a \in A_1(vw)$, it follows that $\sigma(a_1(vw)) = \sigma(v)$.  Analogously, $\sigma(a_i(vw)) = \sigma(v)$ for every \range{i}{1}{4}.  But then, $\{a_i(vw) \mid 1 \leq i \leq 4\}$ is a monochromatic maximal biclique that contains a maximal star, a contradiction.  Therefore, $\sigma(v) \neq \sigma(w)$ for every edge $vw$ of $V(G)$, which implies that the restriction of $\sigma$ to $V(G)$ is a vertex $M$-coloring of $G$.
\end{proof}

Let $G$ be a $K_3$-free graph with no false twins, and define $H$ as the $K_3$-free graph that is obtained from $G$ as in Theorem~\ref{thm:co-bipartite choosability}.  By fixing the list assignment that maps each vertex to $\{1, \ldots, k\}$ in the proof of Theorem~\ref{thm:co-bipartite choosability}, it can be observed that $G$ admits a vertex $k$-coloring if and only if $\overline{H}$ admits a star (resp.\ biclique) $k$-coloring, for every $k \geq 3$.  The problem of determining if a connected $K_3$-free graph with no false twins admits a vertex $k$-coloring is known to be \NP-complete~\cite{Lovasz1973,MaffrayPreissmannDM1996}.  Hence, the star and biclique $k$-coloring problems are $\NP$-complete when restricted to $\overline{K_3}$-free graphs, for every $k \geq 3$.

\begin{theorem}\label{thm:co-k3 coloring}
 \stcol{k} and \bccol{k} are \NP-complete for every $k \geq 3$ when the input is restricted to $\overline{K_3}$-free graphs.
\end{theorem}

\section{Graphs with restricted diamonds}
\label{sec:diamond-free}

The graph $G$ defined in Theorem~\ref{thm:stcol-2} contains a large number of induced diamonds.  For instance, to force different colors on a pair of vertices $v$ and $w$, a $k$-switcher $C$ connecting $\{v, w\}$ is used.  Such switcher contains $O(k^2)$ diamonds, one for each edge of $C$.  An interesting question is, then, whether induced diamonds can be excluded from Theorem~\ref{thm:stcol-2}.  The answer is no, as we prove in this section that the star coloring problem is \NP-complete for diamond-free graphs.  By taking a deeper look at $G$, it can be noted that every diamond of $G \setminus \{s\}$ has a pair of twin vertices.  In order to prove that the star coloring problem is \NP-complete for diamond-free graphs, we show that the problem is \NP even for the larger class of graphs in which every diamond has two twin vertices.  This class corresponds to the class of \{$W_4$, dart, gem\}-free graphs (cf.\ below), and it is worth to note that its graphs may admit an exponential number of maximal stars.  We also study the biclique coloring problem on this class, for which we prove that the problem is \NP when there are no induced $K_{i,i}$ for $i \in O(1)$.  At the end of the section, we study the star and biclique choosability problems, which turn to be \ptwop-hard for \{$C_4$, dart, gem\}-free graph.

Let $G$ be a graph.  Say that $v \in V(G)$ is \Definition{block separable} if every pair of adjacent vertices $w,z\in N(v)$ not dominating $v$ are twins in $G[N(v)]$.  The following lemma shows that \{$W_4$, dart, gem\}-free graphs are precisely those graphs in which every induced diamond has twin vertices, and they also correspond to those graphs is which every vertex is block separable.  This last condition is crucial in the \NP coloring algorithms.

\begin{theorem}\label{thm:w4dartgemequivalence}
The following statements are equivalent for a graph $G$.
  \begin{enumerate}[(i)]
    \item $G$ is \{$W_4$, dart, gem\}-free.
    \item Every induced diamond of $G$ contains a pair of twin vertices.
    \item Every $v \in V(G)$ is block separable. 
  \end{enumerate}
\end{theorem}
\begin{proof} 
 (i) $\Longrightarrow$ (ii) If $D \subseteq V(G)$ induces a diamond with universal vertices $v,w$ and there exists $x \in N(v) \setminus N(w)$, then $D \cup \{x\}$ induces a $W_4$, a dart, or a gem in $G$ depending on the remaining adjacencies between $x$ and the vertices of $D$.

 (ii) $\Longrightarrow$ (iii) Suppose $v \in V(G)$ is not block separable, thus $N[v]$ contains two adjacent vertices $w$ and $z$ not dominating $v$ that are not twins in $G[N[v]]$; say $d(z) \geq d(w)$. Then, $v$ and $z$ are the universal vertices of a diamond containing $w$ and a vertex in $N(z) \setminus N(w)$, i.e., $G$ contains an induced diamond with no twin vertices.

 (iii) $\Longrightarrow$ (i) The $W_4$, dart, and gem graphs have a vertex of degree $4$ that is not block separable.
\end{proof}

Note that if $v$ is block separable, then $N[v]$ can be partitioned into sets $B_0, \ldots, B_\ell$ where $v \in B_0$ and each $B_i$ is a block of $G[N[v]]$.  Moreover, no vertex in $B_i$ is adjacent to a vertex in $B_j$, for $1 \leq i < j \leq \ell$.  We refer to $B_0, \ldots, B_\ell$ as the \Definition{block separation} of $v$.  By definition, $\{v\}S$ is a maximal star of $G$ with $|S|> 1$ if and only if $\ell > 1$, $|S \cap B_0| = 0$ and $|S \cap B_i| = 1$ for \range{i}{1}{\ell}.  By Theorem~\ref{thm:w4dartgemequivalence}, every vertex of a \{$W_4$, dart, gem\}-free graph admits a block separation, hence the next result follows.

\begin{lemma}
 Let $G$ be a \{$W_4$, dart, gem\}-free graph with a coloring $\rho$.  Then, $\rho$ is a star coloring of $G$ if and only if
 \begin{itemize}
  \item $|\rho(B)| = |B|$ for every block $B$ of $G$, and
  \item for every $v \in V(G)$ with block separation $B_0, \ldots, B_\ell$, there exists $B_i$ such that $\rho(v) \not\in \rho(B_i)$.
 \end{itemize}
\end{lemma}

It is well known that the blocks of a graph $G$ can be computed in $O(n+m)$ time.  Hence, it takes $O(d(v)^2)$ time obtain the block separation of a block separable vertex $v$, and, consequently, the star $k$-coloring and the star $k$-choosability problems are in \NP and \ptwop for \{$W_4$, dart, gem\}-free graphs, respectively.

\begin{theorem}\label{thm:np w4-dart-gem}
 \stcol{k} is \NP when the input is restricted to \{$W_4$, dart, gem\}-free graphs.
\end{theorem}

\begin{theorem}\label{thm:ptp w4-dart-gem}
 \stchose{k} is \ptwop when the input is restricted to \{$W_4$, dart, gem\}-free graphs.
\end{theorem}

We now consider the biclique coloring problem.  The algorithm for determining if a coloring $\rho$ is a biclique coloring of $G$ is divided in two steps.  First, it checks that no monochromatic maximal star is a maximal biclique.  Then, it checks that $G$ contains no monochromatic maximal biclique $K_{i,j}$ with $2 \leq i \leq j$. 

For the first step, suppose $\rho$ is a coloring of $G$ where $\rho(B) = |B|$ for every block $B$ of $G$.  Let $v$ be a vertex with a block separation $B_0, \ldots, B_\ell$.  As discussed above, $\{v\}S$ is a maximal star if and only if $\ell > 1$, $|S \cap B_0| = 0$, and $|S \cap B_i| = 1$ for every \range{i}{1}{\ell}.  If $\{v\}S$ is not a maximal biclique, then there exists $w \in V(G) \setminus N[v]$  adjacent to all the vertices in $S$.  Observe that $w$ has at most one neighbor in $B_i$ with color $c$, for each color $c$.  Otherwise, taking into account that twin vertices have different colors, $v,w,y,z$ would induce a diamond with no twin vertices, for $y,z \in N(v) \cap N(w)$.  Therefore, at most one monochromatic maximal star with center $v$ is included in a biclique containing $w$, for each $w \in V(G) \setminus N(v)$.  Thus, to check if there is a monochromatic maximal biclique containing $v$ we first check whether $\prod_{i=1}^\ell|\{z \in B_i \mid \rho(z) = \rho(v)\}| < n$.  If negative, then $\rho$ is not a biclique coloring of $G$.  Otherwise, all the monochromatic maximal stars with center in $v$ are generated in polynomial time, and for each such star $\{v\}S$ it is tested if there exists $w \in V(G) \setminus N[v]$ adjacent to all the vertices in $S$.

\begin{lemma}\label{lem:w4-dart-gem stars}
 If a \{$W_4$, dart, gem\}-free graph $G$ and a coloring $\rho$ are given as input, then it takes polynomial time to determine if there exists a monochromatic maximal biclique $\{v\}S$ with $v \in V(G)$.
\end{lemma}

For the second step, suppose $S$ is an independent set with at least two vertices, and let $I = \bigcap_{v \in S}N(v)$.  Note that if $w,z \in I$ are adjacent, then they are twins in $G$ because $w, z$ are the universal vertices of any induced diamond formed by taking a pair of vertices in $S$.  Hence, $I$ can be partitioned into a collection $B_1, \ldots, B_\ell$ of blocks of $G$ where no vertex in $B_i$ is adjacent to a vertex in $B_j$, for $1 \leq i < j \leq \ell$.  Thus, $ST$ is a maximal biclique of $G$ if and only if no vertex of $V(G) \setminus S$ is complete to $I$ and $|T \cap B_i| = 1$ for every \range{i}{1}{\ell}.  That is, $G$ has a monochromatic maximal biclique $ST$ if and only if $|\rho(S)| = 1$, each block of $I$ has a vertex of color $\rho(S)$, and $\bigcap_{w\in I} N(w) = S$.

\begin{lemma}\label{lem:w4-dart-gem bicliques}
 Let $G$ be a \{$W_4$, dart, gem\}-free graph.  If an independent set $S$ and a coloring $\rho$ of $G$ are given as input, then it takes polynomial time to determine if $G$ has a monochromatic maximal biclique $ST$ with $T \subseteq V(G)$.
\end{lemma}

If $G$ is $K_{i,i}$-free for some constant $i$, then every biclique $ST$ of $G$ with $|S| \leq |T|$ has $|S| < i$.  Thus, to determine if $\rho$ is a biclique coloring of $G$, it is enough traverse every independent set $S$ of $G$ with $O(i)$ vertices and to check that there exists no $T \subset V(G)$ such that $ST$ is a monochromatic maximal biclique.  By Lemmas \ref{lem:w4-dart-gem stars}~and~\ref{lem:w4-dart-gem bicliques}, it takes polynomial time to determine if there exists $T$ such that $ST$ is a monochromatic maximal star.  Since there are $n^{O(1)}$ independent sets with at most $i$ vertices, the algorithm requires polynomial time.  We thus conclude that \bccol{k} and \bcchose{k} are respectively \NP and \ptwop when the input is restricted to \{$K_{i,i}$, $W_4$, dart, gem\}-free graphs.

\begin{theorem}\label{thm:np w4-dart-gem-kii}
 \bccol{k} is \NP when the input is restricted to \{$K_{i,i}$, $W_4$, dart, gem\}-free graphs, for $i\in O(1)$.
\end{theorem}

\begin{theorem}
 \bcchose{k} is \ptwop when the input is restricted to \{$K_{i,i}$, $W_4$, dart, gem\}-free graphs, for $i\in O(1)$.
\end{theorem}

In the rest of this section, we discuss the completeness of the star and biclique coloring and choosability problems.  As in Sections \ref{sec:general case}~and~\ref{sec:choosability}, only one proof is used for each problem because $C_4$-free graphs are considered.  For the reductions, two restricted satisfiability problems are required, namely \naesat and \naesatt.  A valuation $\nu$ of a CNF formula $\phi$ is a \Definition{nae-valuation} when all the clauses of $\phi$ have a true and a false literal.  The formula $(\exists \VEC x)\phi(\VEC x)$ is \Definition{nae-true} when $\phi(\VEC x)$ admits a nae-valuation, while $(\forall \VEC x)\phi(\VEC x)$ is \Definition{nae-true} when every valuation of $\phi(\VEC x)$ is a \Definition{nae-valuation}.   \naesat is the \NP-complete problem (see~\cite{Garey1979}) in which a CNF formula $\phi$ is given, and the goal is to determine if $\phi$ admits a nae-valuation.  Analogously, \naesatt is the \ptwop-complete problem (see~\cite{EiterGottlob1995}) in which a CNF formula $\phi(\VEC x, \VEC y)$ is given, and the purpose is to determine if $(\forall \VEC x)(\exists \VEC y)\phi(\VEC x, \VEC y)$ is nae-true.  We begin discussing the completeness of the star coloring problem.  In order to avoid induced diamonds, we define a replacement of long switchers.

\begin{defn}[diamond $k$-switcher]\label{def:diamond k-switcher}
  Let $G$ be a graph and $U = \{u_1, \ldots, u_h\}$ be an independent set of $G$ with $h \geq 2$.  Say that $S \subset V(G)$ is a \Definition{diamond $k$-switcher connecting $U$} ($k \geq 2$) when $S$ can be partitioned into a vertex $w_1$, a set of leafed vertices $\{w_2, \ldots, w_h\}$, a family $Q_1, \ldots, Q_h$ of $k$-keepers, and a clique $C$ with $k$ vertices in such a way that $C \cup \{w_1\}$ is a clique, $\{w_1\}\{w_2, \ldots, w_h\}$ is a star, $Q_i$ connects $u_i, w_i$ for $\range{i}{1}{h}$, and there are no more edges adjacent to vertices in $S$.
\end{defn}

A diamond $k$-switcher is depicted in Figure~\ref{fig:diamond switcher}.  The main properties of diamond switchers are given in the next lemma.

\begin{figure}[htb]
 \centering
 \includegraphics{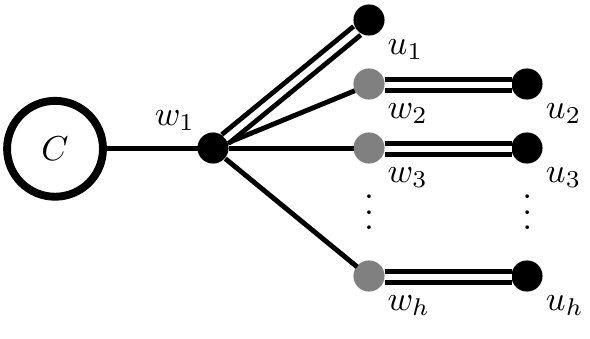}
 \caption{A diamond $k$-switcher connecting $\{u_1, \ldots, u_h\}$.}\label{fig:diamond switcher}
\end{figure}

\begin{lemma}\label{lem:diamond switcher properties}
 Let $G$ be a graph and $S$ be a diamond $k$-switcher connecting $U \subset V(G)$ ($k \geq 2$). Then,
\begin{enumerate}[(i)]
  \item no induced $C_4$, diamond, or $K_{k+2}$ of $G$ contains a vertex of $S$,\label{lem:diamond switcher properties:forbidden}
  \item $|\rho(U)| \geq 2$ for any star $k$-coloring $\rho$ of $G$, and \label{lem:diamond switcher properties:colors}
  \item Any $k$-coloring $\rho$ of $G \setminus S$ in which $|\rho(U)| \geq 2$ can be extended into a $k$-coloring of $G$ in such a way that no monochromatic maximal star has its center in $S$.\label{lem:diamond switcher properties:extension}
\end{enumerate}
\end{lemma}

\begin{proof}
Let $U = \{u_1, \ldots, u_h\}$, $w_1, \ldots, w_h$, $Q_1$, \ldots, $Q_{h}$, and $C$ be as in Definition~\ref{def:diamond k-switcher}.  Statement (\ref{lem:diamond switcher properties:forbidden}) follows by statement~(\ref{lem:keeper properties:forbidden}) of Lemma~\ref{lem:keeper properties}, observing that $U$ is an independent set and that no induced diamond can contain a vertex in a $k$-keeper.  

(\ref{lem:keeper properties:colors}) Let $\rho$ be a star $k$-coloring of $G$.  Since $C$ is a block of size $k$, it contains a vertex $c$ with color $\rho(w_1)$ by Observation~\ref{obs:block coloring}.  Then, taking into account that $\{w_1\}\{c, u_1, w_2, \ldots, w_h\}$ is a maximal star and $\rho(u_i) = \rho(w_i)$ by statement~(\ref{lem:keeper properties:colors}) of Lemma~\ref{lem:keeper properties} for $\range{i}{1}{h}$, it follows that $\rho(u_i) \neq \rho(u_1)$ for some $\range{i}{2}{h}$.

(\ref{lem:keeper properties:extension}) To extend $\rho$, first set $\rho(C) = \{1, \ldots, k\}$ and $\rho(w_i) = \rho(u_i)$ ($\range{i}{1}{h}$), and then iteratively extend $\rho$ to color the leaves and the $k$-keepers according to Observation~\ref{obs:leafed vertex} and statement~(\ref{lem:keeper properties:extension}) of Lemma~\ref{lem:keeper properties}.
\end{proof}

We are now ready to prove the \NP-completeness of the star-coloring problem.

\begin{theorem}\label{thm:w4dartgemnpc}
\stcol{k} is $\NP$-complete when the input is restricted to \{$C_4$, diamond, $K_{k+2}$\}-free graphs for every $k \geq 2$.
\end{theorem}

\begin{proof}
 By Theorem~\ref{thm:np w4-dart-gem}, \stcol{k} is \NP for \{$C_4$, diamond, $K_{k+2}$\}-free graphs.  For the hardness part, we show a polynomial time reduction from \naesat.  That is, given a CNF formula $\phi$ with $\ell$ clauses $\VP P_1,\ldots, \VP P_\ell$ and $n$ variables $\VP x_1, \ldots, \VP x_n$, we define a \{$C_4$, diamond, $K_{k+2}$\}-free graph $G$ such that $\phi$ admits a nae-valuation if and only if $G$ admits a star $k$-coloring.  

 The vertices of $G$ are divided into \Definition{connection} and \Definition{inner} vertices.  For each \range{i}{1}{n} there are two connection vertices $x_i, -x_i$ \Definition{representing} the literals $\VP x_i$ and $\overline{\VP x_i}$, respectively.  Also, there are $k-2$ connection vertices $y_3, \ldots, y_k$.  Let $X = \{x_1, \ldots, x_n, -x_1, \ldots, -x_n\}$, $Y = \{y_3, \ldots, y_k\}$, and $P_h = \{x \in X \mid x \text{ represents a literal in } \VP P_h\}$ for \range{h}{1}{\ell}.  Inner vertices are the vertices included in diamond $k$-switchers connecting connection vertices.  For each $v \in X \cup Y$ and each $y \in Y$ there is a \Definition{color} diamond $k$-switcher connecting $\{v,y\}$.  Also, for each $\range{i}{1}{n}$ there is a \Definition{valuation} diamond $k$-switcher connecting $\{x_i, -x_i\}$.  Finally, there is a \Definition{clause} diamond $k$-switcher connecting $P_h$ for every \range{h}{1}{\ell}.  Observe that $X \cup Y$ is an independent set of $G$.  Thus, by statement~(\ref{lem:diamond switcher properties:forbidden}) of Lemma~\ref{lem:diamond switcher properties}, $G$ is \{$C_4$, diamond, $K_{k+2}$\}-free.  

 Suppose $\phi$ has a nae-valuation $\nu:\VEC x \to \{0,1\}$, and let $\rho$ be a $k$-coloring of the connection vertices such that $\rho(x_i) = 2-\nu(\VP x_i)$ and $\rho(y_j) = j$ for \range{i}{1}{n} and \range{j}{3}{k}.  Clearly, every color or valuation $k$-switcher connects a pair of vertices that have different colors.  Also, since $\nu$ is a nae-valuation, every set $P_h$ (\range{h}{1}{\ell}) has two vertices representing literals $\VP l_1$ and $\VP l_2$ of $\VP P_h$ with $\nu(l_1) \neq \nu(l_2)$.  Hence, $P_h$ is not monochromatic, thus every clause diamond $k$-switcher connects a non-monochromatic set of vertices.  Therefore, by statement~(\ref{lem:diamond switcher properties:extension}) of Lemma~\ref{lem:diamond switcher properties}, $\rho$ can be iteratively extended into a star $k$-coloring of $G$.

 For the converse, suppose $G$ admits a star $k$-coloring $\rho$.  By applying statement~(\ref{lem:diamond switcher properties:colors}) of Lemma~\ref{lem:diamond switcher properties} while considering the different kinds of diamond $k$-switchers, we observe the following facts.  First, by the color diamond $k$-switchers, $|\rho(Y)| = k-2$ and $\rho(X) \cap \rho(Y) = \emptyset$.  Then, we can assume that $\rho(X) \subseteq \{1,2\}$ and $\rho(Y) = \{3, \ldots, k\}$.  Hence, by the valuation diamond $k$-switcher connections, we obtain that $\rho(x_i) \neq \rho(-x_i)$ for every $\range{i}{1}{n}$.  Thus, the mapping $\nu:\VEC x \to \{0,1\}$ such that $\nu(\VP x_i) = 2-\rho(x_i)$ is a valuation.  Moreover, by the clause diamond $k$-switcher connections, $P_h$ is not monochromatic for \range{h}{1}{\ell}.  Consequently, $\nu$ is a nae-valuation of $\phi$.
\end{proof}

Observe that the graph $G$ defined in Theorem~\ref{thm:w4dartgemnpc} is not chordal.  However, as discussed in Section~\ref{sec:chordal}, every edge $xy$ such that $x,y$ are connected by a $k$-keeper (inside the diamond $k$-switchers) can be subdivided so as to eliminate all the induced holes of length at most $i$, for every $i \in O(1)$.

We now deal with the star choosability problem.  Recall that long switchers are not well suited for the star choosability problem because they contain keepers, and vertices connected by keepers need not have the same colors in every list coloring.  Keepers are also present inside diamond switchers, thus it is not a surprise that diamond $k$-switchers are not star $k$-choosable.   For this reason, as in Section~\ref{sec:choosability}, the proof is by induction, using list switchers for $k=2$.  Since list switchers contain induced diamonds, the \stp-hardness will be obtained for \{$W_4$, gem, dart\}-free graphs, and not for diamond-free graphs.  Unfortunately, we did not find a way to avoid these diamonds.  Moreover, some kind of forcers are required as well; our forcers have induced diamonds that we were not able to remove either.  The hardness proof for $k=2$ is, in some sense, a combination of the proofs of Theorems~\ref{thm:stchose-2}~and~\ref{thm:w4dartgemnpc}.  Roughly speaking, the idea is to force the colors of the universal variables of the input formula as in Theorem~\ref{thm:stchose-2}, while a nae-valuation is encoded with colors as in Theorem~\ref{thm:w4dartgemnpc}.

\begin{theorem}\label{thm:w4-dart-gem ptp complete}
  \stchose{2} is \ptwop-hard when its input is restricted to \{$C_4$, dart, gem, $K_4$\}-free graphs.
\end{theorem}

\begin{proof}
 The hardness of \stchose{2} is obtained by reducing \naesatt.  That is, given a CNF formula $\phi(\VEC z, \VEC x)$ with $\ell$ clauses $\VP P_1, \ldots, \VP P_\ell$, and $m+n$ variables $\VEC x= \VP x_1, \ldots, \VP x_n$, $\VEC z = \VP z_1, \ldots, \VP z_m$, we build a \{$C_4$, dart, gem, $K_4$\}-free graph $G$ that is star $2$-choosable if and only if $(\forall\VEC z)(\exists \VEC x)\phi(\VEC z, \VEC x)$ is nae-true.  For the sake of simplicity, in this proof we use $i$, $f$, and $h$ as indices that refer to values in $\{1, \ldots, n\}$, $\{1, \ldots, m\}$, and $\{1, \ldots, \ell\}$.
  
 Graph $G$ is an extension of the graph in Theorem~\ref{thm:w4dartgemnpc} for $k = 2$ (replacing diamond switcher with list switchers).  It has a \Definition{connection} vertex $x_i$ (resp.\ $-x_i$, $z_f$, $-z_f$) \Definition{representing} $\VP x_i$ (resp.\ $\overline{\VP x_i}$, $\VP z_f$, $\overline{\VP z_f}$) for each $i$ (and each $f$), and two \Definition{connection} vertices $t$, $-t$.  Let $L_X = \{x_i, -x_i \mid 1 \leq i \leq n\}$, $Z = \{z_1, \ldots, z_m\}$, $-Z = \{-z_1, \ldots, -z_m\}$, and $P_h = \{x \mid x \in X \cup Z \cup -Z \text{ represents a literal in } \VP P_h\}$.  Graph $G$ also has \Definition{inner vertices} which are the vertices in list switchers and $2$-forcers connecting connection vertices.  There are $2$-forcers connecting each vertex of $Z \cup \{t, -t\}$, and list switchers connecting: $\{t, x_i, -x_i\}$ and $\{-t, x_i, -x_i\}$ for each $i$; $\{z_f,-z_f\}$ for each $f$; and $P_h \cup \{t\}$ and $P_h \cup \{-t\}$ for each $h$.  
  
 Let $L$ be a $2$-list assignment of $G$, and suppose $(\forall\VEC z)(\exists\VEC x)\phi(\VEC z, \VEC x)$ is nae-true.  Define $\rho$ as an $L$-coloring of the connection vertices satisfying the following conditions.
 \begin{enumerate}[(i)]
  \item $\rho(v)$ is any color $L$-admissible for $v \in Z \cup \{t,-t\}$.  Such a color always exists by statement~(\ref{lem:forcer properties:extension}) of Lemma~\ref{lem:forcer properties}.  Suppose, w.l.o.g., that $\rho(t) = 1$ and $\rho(-t) \in \{0,1\}$, and define $\nu(\VEC z)$ as a valuation of $\VEC z$ such that $\nu(\VP z_f) = 1$ if and only if $\rho(z_f) = 1$.  

  \item By hypothesis, $\nu$ can be extended into a nae-valuation of $\phi(\VEC z, \VEC x)$.  If $\nu(\VP x_i) \in L(x_i)$, then $\rho(x_i) = \nu(\VP x_i)$.  Otherwise, $\rho(x_i) \in L(x_i) \setminus \{1 - \nu(\VP x_i)\}$.  Similarly, $\rho(-x_i) = 1-\nu(\VP x_i)$ if $1-\nu(\VP x_i) \in L(-x_i)$, while $\rho(-x_i) \in L(-x_i) \setminus \{\nu(\VP x_i)\}$ otherwise.
    
  \item $\rho(-z_f) \in L(z_f) \setminus \{\rho(z_f)\}$.  
 \end{enumerate}
 It is not hard to see that $\rho$ can always be obtained.  Observe that $\rho(-z_f) \neq \rho(z_f)$, $\rho(x_i) = \rho(-x_i)$ only if $\rho(x_i), \rho(-x_i) \not\in \{\rho(t), \rho(-t)\}$, and $\rho(P_h)$ is monochromatic only if $\rho(P_h) \not\subseteq \{\rho(t), \rho(-t)\}$.  Therefore, $\rho$ can be extended into an $L$-coloring of $G$ by statement~(\ref{lem:list switcher properties:extension}) of Lemma~\ref{lem:list switcher properties} and statement~(\ref{lem:forcer properties:extension}) of Lemma~\ref{lem:forcer properties}.
 
 For the converse, suppose $G$ is star $2$-choosable, and consider any valuation $\nu$ of $\VEC z$.  Define $L$ to be a $2$-list assignment of $G$ such that $1$ is the unique color admissible for $t$, $0$ is the unique color admissible for $-t$, $\nu(\VP z_f)$ is the unique color admissible for $z_f$, and $L(v) = \{0,1\}$ for $v \in V(G) \setminus (Z \cup \{t,-t\})$.  By statement~(\ref{lem:forcer properties:color}) of Lemma~\ref{lem:forcer properties}, such a list assignment always exists.  Let $\rho$ be a star $L$-coloring of $G$.  By repeatedly applying statement~(\ref{lem:list switcher properties:colors}) of Lemma~\ref{lem:list switcher properties}, it can be observed that none of $\{t,x_i,-x_i\}$, $\{-t, x_i, -x_{i}\}$, $\{t\} \cup P_h$, and $\{-t\}$ are monochromatic.  Therefore, $\nu$ is a nae-valuation of $\phi(\VEC z, \VEC x)$.
\end{proof}

Note that if $C$ is a $k$-switcher, then every induced diamond that contains a vertex $u \in C$ also contains a twin of $u$.  Hence, by Theorem~\ref{thm:w4dartgemequivalence}, if $F$ is a $k$-forcer connecting $v$, then no vertex in $F$ belongs to an induced dart or gem.  Consequently, if $G_k$ is a \{$C_4$, dart, gem, $K_{k+2}$\}-free graph, then the graph $G_{k+1}$ defined in the proof of Theorem~\ref{thm:stchose-k} is \{$C_4$, dart, gem, $K_{k+3}$\}-free.  That is, a verbatim copy of the proof of Theorem~\ref{thm:stchose-k} can be used to conclude the following.

\begin{theorem}
  \stchose{k} is \ptwop-complete when its input is restricted to \{$C_4$, dart, gem, $K_{k+2}$\}-free graphs for every $k \geq 2$.
\end{theorem}

\section{Split graphs}
\label{sec:split}

In this section we consider the star coloring and star choosability problems restricted to split graphs.  The reason for studying split graphs is that they form an important subclass of chordal graphs, and also correspond to the class of \{$2K_2$, $C_4$, $C_5$\}-free graphs~\cite{Golumbic2004}.  A graph $G$ is \Definition{split} when its vertex set can be partitioned into an independent set $S(G)$ and a clique $Q(G)$.  There are $O(d(v))$ maximal stars centered at $v \in V(G)$, namely $\{v\}(N(v) \cap S(G))$ and $\{v\}((N(v) \cup \{w\})\setminus N(w))$ for $w \in Q(G)$.  Thus, the star coloring and star choosability problems on split graphs are \NP and \ptwop, respectively.  In this section we prove the completeness of both problems.

We begin observing that $G$ admits a star coloring with $\beta+1$ colors, where $\beta$ is the size of the maximum block.  Indeed, each block $B$ of $Q(G)$ is colored with colors $\{1, \ldots, |B|\}$, while each vertex of $S(G)$ is colored with color $\beta+1$.  We record this fact in the following observation.
\begin{observation}
 If $G$ is a split graph whose blocks have size at most $\beta$, then $G$ admits a star coloring using $\beta+1$ colors.  Furthermore, such a coloring can be obtained in linear time.
\end{observation}
Computing a star coloring of a split graph using $\beta+1$ colors is easy, but determining if $\beta$ colors suffice is an $\NP$-complete problem.  The proof of hardness is almost identical to the one in Section~\ref{sec:diamond-free} for \{$C_4$, diamond, $K_{k+2}$\}-free graphs.  That is, given a CNF formula $\phi$ we build a graph $G$ using \emph{split} switchers in such a way that $\phi$ admits a nae-valuation if and only if $G$ admits a star $k$-coloring.  As switchers, \Definition{split} switchers force a set of vertices to have at least two colors in a star $k$-coloring.  The difference is that split switchers do so in a split graph (see Figure~\ref{fig:split-switcher}).

\begin{figure}
 \centering\includegraphics{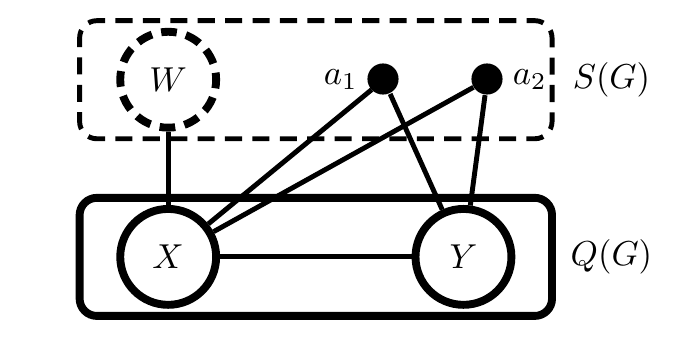}
 \caption{A split $k$-switcher connecting a set $W$.}\label{fig:split-switcher}
\end{figure}

\begin{defn}[split $k$-switcher]\label{def:nae connection}
  Let $G$ be a split graph and $W \subseteq S(G)$ with $|W| \geq 2$.  Say that $S \subset V(G) \setminus W$ is a \Definition{split $k$-switcher connecting $W$} when $S$ can be partitioned into two sets $X, Y \subseteq Q(G)$ and two vertices $a_1, a_2 \in S(G)$ in such a way that $|X| = |Y| = k$, $X \cup \{v\}$ is a clique for every $v \in W \cup \{a_1, a_2\}$, $Y \cup \{a_1\}$ and $Y \cup \{a_2\}$ are cliques, there are no more edges between vertices in $X \cup Y$ and vertices in $S(G)$, and there are no more edges incident to $a_1$ and $a_2$.
\end{defn}

The properties of split $k$-switchers are summarized in the following lemma.

\begin{lemma}\label{lem:split switcher properties}
 Let $G$ be a split graph and $S$ be a split $k$-switcher connecting $W \subset S(G)$. Then,
\begin{enumerate}[(i)]
  \item $|\rho(W)| > 1$ for every star $k$-coloring $\rho$ of $G$, and\label{lem:split switcher properties:colors}
  \item For every $k$-list assignment of $G$, any $L$-coloring $\rho$ of $G \setminus S$ in which $|\rho(W)| > 1$ can be extended into an $L$-coloring of $G$ in such a way that no monochromatic maximal star has its center in $S$.\label{lem:split switcher properties:extension}
\end{enumerate}
\end{lemma}

\begin{proof}
 Let $X$, $Y$, $a_1$, and $a_2$ be as in Definition~\ref{def:nae connection}.
 
 (\ref{lem:split switcher properties:colors}) Let $\rho(G)$ be a star $k$-coloring of $G$.  By definition, $X$ and $Y$ are blocks of $G$, thus $|\rho(X)| = |\rho(Y)| = k$ by Observation~\ref{obs:block coloring}.  Let $x_c$ and $y_c$ be the vertices with color $c$ in $X$ and $Y$, respectively.  Since the maximal star $\{x_c\}(W \cup \{y_c\})$ is not monochromatic, it follows that $\rho(W) \neq \{c\}$.  Therefore, $|\rho(W)| > 1$.
 
 (\ref{lem:split switcher properties:extension}) To extend $\rho$ to $S$, define $|\rho(X)| = |\rho(Y)| = k$, and $\rho(a_1) \neq \rho(a_2)$.  Let $\{v\}V$ be a maximal star with $v \in S$.  If either $v \in Y$ or $v \in X$ and $V \cap Y = \emptyset$, then $\{a_1, a_2\} \subseteq V$, thus $\{v\}V$ is not monochromatic.  Otherwise, if $v \in X$ and $V \cap Y \neq \emptyset$, then $W \subseteq V$, thus $\{v\}V$ is not monochromatic as well.
\end{proof}

By replacing diamond $k$-switchers with split $k$-switchers in the proof of Theorem~\ref{thm:w4dartgemnpc}, the $\NP$-completeness of the star $k$-coloring problem for split graphs is obtained.  For the sake of completeness, we sketch the proof, showing how to build the graph $G$ from the CNF formula.

\begin{theorem}\label{thm:splitnpc}
\stcol{k} restricted to split graphs is $\NP$-complete for every $k \geq 2$. 
\end{theorem}

\begin{proof}
 Split graphs have $O(n+m)$ maximal stars, hence \stcol{k} is $\NP$ for split graphs.  For the hardness part, let $\phi$ be a CNF formula with $m$ clauses $\VP P_1,\ldots, \VP P_m$ and $n$ variables $\VP x_1, \ldots, \VP x_n$.  Define $G$ as the split graphs with \Definition{connection} and \Definition{inner} vertices, as follows.  For each \range{i}{1}{n} there are two connection vertices $x_i, -x_i$ \Definition{representing} the literals $\VP x_i$ and $\overline{\VP x_i}$, respectively.  Also, there are $k-2$ connection vertices $y_3, \ldots, y_k$.  Let $X = \{x_1, \ldots, x_n, -x_1, \ldots, -x_n\}$, $Y = \{y_3, \ldots, y_k\}$, and $P_h = \{x \in X \mid x \text{ represents a literal in } \VP P_h\}$ for \range{h}{1}{\ell}.  Inner vertices form the split $k$-switchers connecting connection vertices.  For each $v \in X \cup Y$ and each $y \in Y$ there is a \Definition{color} split $k$-switcher connecting $\{v,y\}$.  Also, for each $\range{i}{1}{n}$ there is a \Definition{valuation} split $k$-switcher connecting $\{x_i, -x_i\}$.  Finally, there is a \Definition{clause} split $k$-switcher connecting $P_h$ for every \range{h}{1}{\ell}. Clearly, $G$ is a split graph, and, as in the proof of Theorem~\ref{thm:w4dartgemnpc}, $G$ admits star $k$-star coloring if and only if $\phi$ admits a nae-valuation.
\end{proof}

For the star choosability problem of split graphs, the idea is to adapt the proof of Theorem~\ref{thm:w4-dart-gem ptp complete}, providing a new kind of forcer.  This new forcer is called the \emph{split $k$-forcer} and it is just a $k$-forcer where its $k$-switchers form a clique.  For the sake of completeness, we include its definition.

\begin{defn}[split $k$-forcer]\label{def:k-split-forcer connection}
  Let $G$ be a split graph and $v \in S(G)$.  Say that $F \subseteq V(G)$ is a \emph{split $k$-forcer connecting $v$} ($k \geq 2$) when $F$ can be partitioned into sets $A, B \subseteq S(G)$ and $C(a,b) \subseteq Q(G)$ for $a \in A\cup\{v\}$ and $b \in B$ in such a way that $|A| = k-1$, $|B| = k^k-1$, $C(a, b) \cup \{a\}$ and $C(a,b) \cup \{b\}$ are cliques, there are no more edges between vertices in $F \cap Q(G)$ and vertices in $S(G)$, and there are no more edges incident to vertices in $A \cup B$.
\end{defn}

Let $L$ be a $k$-list assignment of a split graph $G$ and $F$ be a split $k$-forcer connecting $v \in V(G)$.  As in Section~\ref{sec:choosability}, we say that $c \in L(v)$ is \emph{$L$-admissible for $v$} when there is an $L$-coloring $\rho$ of $G$ such that $\rho(v) = c$ and no monochromatic maximal star has its center in $F$.  The following lemma resembles Lemma~\ref{lem:forcer properties}.

\begin{lemma}\label{lem:split forcer properties}
  Let $G$ be a split graph and $F$ be a split $k$-forcer connecting $v \in S(G)$.  Then, 
  \begin{enumerate}[(i)]
    \item for every $k$-list assignment $L$ of $G$ there is an $L$-admissible color for $v$, and\label{lem:split forcer properties:extension}
    \item every $k$-list assignment $L$ of $G \setminus F$ can be extended into a $k$-list assignment of $G$ in which $v$ has a unique $L$-admissible color.
  \end{enumerate}
\end{lemma}

\begin{proof}
  The lemma can be proven with a verbatim copy of the proof of Lemma~\ref{lem:forcer properties}.  In particular, observe that a split $k$-forcer can be obtained from a $k$-forcer by inserting the edges between $C(a,b)$ and $C(a', b')$, for every $a,a' \in A \cup \{v\}$ and $b,b' \in B$.  
\end{proof}

The hardness of the star choosability problem is also obtained by adapting the proof of Theorem~\ref{thm:w4-dart-gem ptp complete}.  We remark, however, that in this case no induction is required, because split $k$-switchers are $k$-choosable by statement~(\ref{lem:split switcher properties:extension}) of Lemma~\ref{lem:split switcher properties}.  The proof is sketched in the following theorem.

\begin{theorem}\label{thm:split stchose}
 \stchose{k} restricted to split graphs is \ptwop-complete for every $k \geq 2$.
\end{theorem}

\begin{proof}
 \stchose{k} is \ptwop for split graphs because split graphs have a polynomial amount of maximal stars.  Let $\phi(\VEC z, \VEC x)$ be a CNF formula with $\ell$ clauses $\VP P_1, \ldots, \VP P_\ell$, and $m+n$ variables $\VEC x= \VP x_1, \ldots, \VP x_n$, $\VEC z = \VP z_1, \ldots, \VP z_m$.  Use $i$, $f$, $h$, and $q$ to denote indices in $\{1, \ldots, n\}$, $\{1, \ldots, m\}$, $\{1, \ldots, \ell\}$, and $\{3, \ldots, k\}$, respectively.  Define $G$ as the split graph that has a \Definition{connection} vertex $x_i \in S(G)$ (resp.\ $-x_i$, $z_f$, $-z_f$) \Definition{representing} $\VP x_i$ (resp.\ $\overline{\VP x_i}$, $\VP z_f$, $\overline{\VP z_f}$), and $k$ \Definition{connection} vertices $t$, $-t$, $y_3, \ldots, y_{k}$.  Let $X = \{x_i, -x_i \mid 1 \leq i \leq n\}$, $Y = \{y_3, \ldots, y_k\}$, $Z = \{z_1, \ldots, z_m\}$, $-Z = \{-z_1, \ldots, -z_m\}$, and $P_h = \{x \mid x \in X \cup Z \cup -Z \text{ represents a literal in } \VP P_h\}$.  Graph $G$ also has \Definition{inner vertices} that are the vertices of split $k$-switchers and split $k$-forcers connecting connection vertices.  There are split $k$-forcers connecting $z_f$, $t$, $-t$, and split $k$-switchers connecting: $\{t, x_i, -x_i\}$, $\{-t, x_i, -x_i\}$, $\{z_f, -z_f\}$, $P_k \cup \{t\}$, $P_k \cup \{-t\}$, and $\{v, y_q\}$ for every $v \in X \cup Y \cup -Z \cup \{t,-t\}$.  

 Following the proof of Theorem~\ref{thm:w4-dart-gem ptp complete}, it can be observed that, for $k = 2$, $G$ admits a star $L$-coloring, for a $k$-list assignment $L$ of $G$, if and only if $(\forall \VEC x)(\exists \VEC y)\phi(\VEC x, \VEC y)$ is nae-true.  For $k > 2$, observe that if $(\forall \VEC x)(\exists \VEC y)\phi(\VEC x, \VEC y)$ is nae-true, then a star $L$-coloring $\rho$ is obtained if $\rho(Y)$ is taken so that $|\rho(Y)|= k-2$ and $\rho(t), \rho(-t) \not\in \rho(Y)$.  Conversely, if $G$ admits a star $L$-coloring for the $k$-list assignment in which $L(v) = \{0, \ldots, k-1\}$ for every connecting vertex $v \not\in Z \cup \{t,-t\}$, then $\rho(X) = \{0,1\}$, thus a nae-valuation of $\phi(\VEC x, \VEC y)$ is obtained.
\end{proof}

To end this section, consider the more general class of $\overline{C_4}$-free graphs.  By definition, $\{v\}S$ is a star of a graph $G$ if and only if $S$ is a maximal independent set of $G[N(v)]$.  In~\cite{FarberDM1989}, it is proved that $G[N(v)]$ has $O(d(v)^2)$ maximal independent sets when $G$ is $\overline{C_4}$-free.  Thus, $\overline{C_4}$-free graphs have $O(nm)$ maximal stars, which implies that the star coloring and star choosability problems on this class are \NP and \ptwop, respectively.  

\begin{theorem}\label{thm:2k2 coloring}
 \stcol{k} and \stchose{k} are respectively \NP-complete and \ptwop-complete for every $k \geq 2$ when the input is restricted to $\overline{C_4}$-free graphs.
\end{theorem}

\section{Threshold graphs}
\label{sec:threshold}

\newcommand{\AsVec}[1]{[#1]}

Threshold graphs form a well studied class of graphs which posses many definitions and characterizations~\cite{Golumbic2004,MahadevPeled1995}.  The reason for studying them in this article is that threshold graphs are those split graphs with no induced $P_4$'s.  Equivalently, a graph is a threshold graph if an only if it is $\{2K_2, P_4, C_4\}$-free.

In this section we develop a linear time algorithm for deciding if a threshold graph $G$ admits a star $k$-coloring.  If affirmative, then a star $k$-coloring of $G$ can be obtained in linear time. If negative, then a certificate indicating why $G$ admits no coloring is obtained.  We prove also that $G$ is star $k$-choosable if and only if $G$ admits a star $k$-coloring.  Thus, deciding whether $G$ is star $k$-choosable takes linear time as well.  It is worth noting that threshold graphs can be encoded with $O(n)$ bits using two sequences of natural numbers (cf.\ below).  We begin this section with some definitions on such sequences.

Let $S = s_1, \ldots, s_r$ be a sequence of natural numbers.  Each \range{i}{1}{r} is called an \Definition{index} of $S$.  For $k \in \mathbb{N}$, we write $S = \AsVec{k}$ and $S \leq \AsVec{k}$ to respectively indicate that $s_i = k$ and $s_i \leq k$ for every index $i$.  Similarly, we write $S > \AsVec{k}$ when $S \not\leq \AsVec{k}$, i.e., when $s_i > k$ for some index $i$.  Note that $S$ could be empty; in such case, $S = \AsVec{k}$ and $S \leq \AsVec{k}$ for every $k \in \mathbb{N}$.  For indices $i, j$, we use $S[i,j]$ to denote the sequence $s_i, \ldots, s_j$.  If $i > j$, then $S[i,j] = \emptyset$.  Similarly, we define $S(i, j] = S[i+1,j]$, $S[i, j) = S[i,j-1]$, and $S(i,j) = S[i+1, j-1]$.  

A \Definition{threshold representation} is a pair $(Q, S)$ of sequences of natural numbers such that $|Q| = |S| + 1$. Let $Q = q_1, \ldots, q_{r+1}$ and $S = s_1, \ldots, s_r$.  Each threshold representation defines a graph $G(Q, S)$ whose vertex set can be partitioned into $r+1$ blocks $Q_1, \ldots, Q_{r+1}$ with $|Q_1| = q_1, \ldots, |Q_{r+1}| = q_{r+1}$ and $r$ independent sets $S_1, \ldots, S_r$ with $|S_1| = s_1, \ldots, |S_r| = s_r$ such that, for $1 \leq i \leq j \leq r$, the vertices in $Q_i$ are adjacent to all the vertices in $S_j \cup Q_{j+1}$.  It is well known that $G$ is a connected threshold graph if and only if it is isomorphic to $G(Q, S)$ for some threshold representation $(Q, S)$~\cite{Golumbic2004,MahadevPeled1995}.  The following observation describes all the maximal stars of $G(Q, S)$.

\begin{observation}\label{obs:stars threshold}
  Let $(Q, S)$ be a threshold representation and $v$ be a vertex of $G(Q, S)$.  Then, $\{v\}W$ is a maximal star of $G(Q, S)$ if and only if there are indices $i \leq j$ of $Q$ such that $v \in Q_i$, and $W = \{w\} \cup \bigcup_{h=i}^{j-1} S_h$ for some vertex $w \in Q_j$.
\end{observation}

For $k \in \mathbb{N}$, we say that index $i$ of $Q$ is \Definition{$k$-forbidden} for the threshold representation $(Q, S)$ when either $q_i > k$ or $q_i = k$ and there exists some index $j > i$ such that $q_j = k$, $Q(i, j) = \AsVec{k-1}$ and $S[i, j) = \AsVec{1}$.  The next theorem shows how to obtain a star $k$-coloring of $G$ when a threshold representation is provided.

\begin{theorem}\label{thm:threshold coloring}
  The following statements are equivalent for a threshold representation $(Q, S)$.
  \begin{enumerate}
    \item $G(Q,S)$ is star $k$-choosable.
    \item $G(Q,S)$ admits a star $k$-coloring.
    \item No index of $Q$ is not $k$-forbidden for $(Q, S)$.
  \end{enumerate}
\end{theorem}

\begin{proof}
  (i) $\Longrightarrow$ (ii) is trivial.  

  (ii) $\Longrightarrow$ (iii).  Suppose $G(Q, S)$ admits a star $k$-coloring $\rho$ and yet $Q$ contains some $k$-forbidden index $i$.  Since $Q_i$ is a block of $G(Q, S)$, then $q_i \leq k$ by Observation~\ref{obs:block coloring}.  Hence, $q_i = k$ and there exists an index $j > i$ such that $q_j = k$, $Q(i,j) = \AsVec{k-1}$, and $S[i,j) = \AsVec{1}$.  Let $w_h$ be the unique vertex in $S_h$ for \range{h}{i}{j-1}.  By Observation~\ref{obs:block coloring}, both $Q_i$ and $Q_j$ have at least one vertex of each color \range{c}{1}{k}, while for each index \range{h}{i+1}{j-1} there exists a color $c_h$ such that $\rho(Q_h) = \{1, \ldots, k\} \setminus \{c_h\}$.  Consequently, there are indices $a < b$ in $\{i, \ldots, j\}$ such that $\rho(w_a) \in \rho(Q_a) \cap \rho(Q_b)$ and $c_h = \rho(w_h) = \rho(w_a)$ for every index \range{h}{a+1}{b-1}.  Indeed, it is enough to take \range{a}{i}{j-1} as the maximum index with $\rho(w_a) \in \rho(Q_a)$ and \range{b}{a}{j} as the minimum index such that $\rho(w_a) \not\in \rho(Q_{a+1}) \cup \ldots \cup \rho(Q_{b-1})$.  Therefore, if $v_a$ and $v_b$ are the vertices of $Q_a$ and $Q_b$ with color $c$, respectively, then $\{v_a\}\{v_b, w_{a+1}, \ldots, w_{b-1}\}$ is a monochromatic maximal star by Observation~\ref{obs:stars threshold}, a contradiction.

  (iii) $\Longrightarrow$ (i).  Let $L$ be a $k$-list assignment of $G(Q,S)$ and define $w_i$ as any vertex of $S_i$ for each index $i$ of $I$.  For each index $i$ of $Q$, define \range{p(i)}{1}{i} as the minimum index such that $Q(p(i), i] < \AsVec{k}$.  Let $\rho$ be an $L$-coloring of $G(Q, S)$ that satisfies all the following conditions for every index $i$ of $Q$:
  \begin{enumerate}[(1)]
    \item $|\rho(Q_i)| = q_i$,
    \item if $q_i < k$ and $i \neq r+1$, then $\rho(w_{p(i)})$ and $\rho(w_i)$ do not belong to $\rho(Q_i)$,
    \item if $q_i < k-1$ and $1 < i < r+1$, then $\rho(w_i) \neq \rho(w_{p(i)})$, and 
    \item if $s_i > 1$ and $i \neq r+1$, then $|\rho(S_i)| \geq 2$.
  \end{enumerate}
  A coloring satisfying all the above conditions can obtained iteratively, by coloring the vertices in $Q_i \cup S_i$ before coloring the vertices in $Q_j \cup S_j$ for every pair of indices $i < j$.  We claim that $\rho$ is a star $L$-coloring of $G(Q, S)$.  To see why, let $\{v\}W$ be a maximal star of $G(Q, S)$.  By Observation~\ref{obs:stars threshold}, there are two indices $i \leq j$ of $Q$ such that $v \in Q_i$ and $W = \{w\} \cup S_i \cup \ldots \cup S_{j-1}$ for some $w \in Q_j$.  If $i = j$, then $\rho(v) \neq \rho(w)$ by~(1).  If $S[i, j) > \AsVec{1}$, then $S_h$ is not monochromatic by~(4).  If $q_i < k$, then $\{v\} \cup S_i$ is not monochromatic by~(2).  If $q_i = k$ and $q_j < k$, then \range{p(i)}{i}{j-1}, thus $\{v\} \cup S_{p(j)} \subset W$ is not monochromatic by~(2).  Finally, if $q_i = q_j = k$ and $S[i, j) = \AsVec{1}$, then there exists index $h$ such that $q_h < k -1$; otherwise $i$ would be a $k$-forbidden index of $(Q, S)$.  Then, by~(3), $\rho(w_h) \neq \rho(w_{p(h)})$ which implies that $W$ is not monochromatic.  Summing up, we conclude that $G(Q,S)$ has no monochromatic maximal star.
\end{proof}

Theorem~\ref{thm:threshold coloring} has several algorithmic consequences for a threshold representation $(Q, S)$ of a graph $G$.  As mentioned, $G$ is a split graph where $Q(G) = \bigcup Q_i$ and $S(G) = \bigcup S_i$, thus $\chi_S(G)$ is either $k$ or $k+1$, for $k = \max(Q)$.  While deciding if $k$ colors suffice for a general split graph is an {\NP}-complete problem, only $O(|Q|)$ time is needed to decide whether $\chi_S(G) = k$ when $(Q, S)$ is given as input; it is enough to find a $k$-forbidden index of $Q$.  Furthermore, if $k < \chi_S(G)$, then a $k$-forbidden index can be obtained in $O(|Q|)$ time as well.  Also, if $\chi_S(G) = k$, then a star $k$-coloring $\rho$ of $G$ can be obtained in linear time by observing rules (1)--(4) of implication (iii) $\Longrightarrow$ (i).  To obtain $\rho$, begin traversing $Q$ to find $p(i)$ for every index $i$ of $Q$.  Then, color the vertices of each block $Q_i$ with colors $1, \ldots, q_i$.  Following, color the vertices $w_1, \ldots, w_r$ in that order, taking the value of $q_i$ into account for each index $i$ of $S$.  If $q_i \geq k-1$, then $\rho(w_i) = k$; otherwise, $\rho(w_i)$ is any value in $\{q_i+1, \ldots, k\} \setminus \rho(w_{p(i)})$.  Finally, color the vertices in $S_i \setminus \{w_i\}$ with color $1$ for each index $i$ of $S$.  To encode $\rho$ only two values are required for each index $i$ of $Q$, namely, $q_i$ and $\rho(w_i)$.  Thus, $\rho$ can be obtained in $O(|Q|)$ time as well.  Finally, $(Q, S)$ can be obtained in $O(n+m)$ time from $G$ when $G$ is encoded with adjacency lists~\cite{Golumbic2004,MahadevPeled1995}.  Thus, all the above algorithms take $O(n+m)$ when applied to the adjacency list representation of $G$.  We record all these observation in the theorem below.

\begin{theorem}
  Let $(Q, S)$ be a threshold representation of a graph $G$.  The problems of computing $\chi_S(G)$, $ch_S(G)$, a $k$-forbidden index of $Q$ for $k < \chi_S(G)$, and a $\chi_S(G)$-coloring of $G$ can be solved in $O(|Q|)$ time when $(Q, S)$ is given as input.  In turn, all these problems take $O(n+m)$ time when an adjacency list representation of $G$ is given.
\end{theorem}

\section{Net-free block graphs}
\label{sec:block}

In this section we study the star coloring and star choosability problems on block graphs.  A graph is a \Definition{block graph} if it is chordal and diamond-free.  We develop a linear time algorithm for deciding if a net-free block graph $G$ admits a star $k$-coloring.  The algorithm is certified; a star $k$-coloring of $G$ is obtained in the affirmative case, while a forbidden induced subgraph of $G$ is obtained in the negative case.  As threshold graphs, net-free block graphs admit an $O(n)$ space representation using weighted trees (cf.\ below).  The certificates provided by the algorithms are encoded in such trees, and can be easily transformed to certificates encoded in terms of $G$.  We begin describing the weighted tree representation of net-free block graphs.

A \Definition{netblock representation} is a pair $(T, K)$ where $T$ is a tree and $K$ is a \Definition{weight function} from $E(T)$ to $\mathbb{N}$.  The graph $G(T, K)$ \Definition{represented} by $(T, K)$ is obtained by inserting a clique $B(vw)$ with $K(vw)$ vertices adjacent to $v$ and $w$, for every $vw \in E(T)$.  It is well known that a connected graph is a netblock graph if and only if it is isomorphic to $G(T, K)$ for some netblock representation $(T, K)$~\cite{BrandstadtLeSpinrad1999}. By definition, two vertices $x, y$ of $G(T, K)$ are twins if and only if 1.\ $xy \in B(vw)$ for some $vw \in E(T)$, or 2.\ $x$ is a leaf of $T$ adjacent to $v \in V(T)$ and $y \in B(vx)$, or 3.\ both $x$ and $y$ are leafs of $T$ (in which case $G(T, K)$ is a complete graph).  The following observation describes the remaining maximal stars of $G(T, K)$.

\begin{observation}
\label{obs:blockbicliques}
Let $(T, K)$ be a netblock representation and $v$ be a vertex of $G(T, K)$.  Then, $\{v\}S$ is a maximal star of $G(T, K)$ with $|S| > 1$ if and only if $v$ is an internal vertex of $T$ and $S$ contains exactly one vertex of $B(vw) \cup \{w\}$ for every $w \in N_T(v)$.
\end{observation}

Let $(T, K)$ be a netblock representation.  For $k \in \mathbb{N}$, a \Definition{$k$-subtree} of $(T, K)$ is a subtree of $T$ formed by edges of weight $k$; a \Definition{maximal} $k$-subtree is a $k$-subtree whose edge set is not included in the edge set of another $k$-subtree.  A \Definition{$k$-exit} vertex of $(T, K)$ is a vertex $v$ that has some neighbor $w \in V(T)$ such that $K(vw) < k$.  The following theorem characterize those net-free block graphs that admit a star $k$-coloring.

\begin{theorem}\label{thm:block graph coloring}
The following statements are equivalent for a netblock representation $(T, K)$.

\begin{enumerate}[i.]
  \item $G(T, K)$ is star $k$-choosable.
  \item $G(T, K)$ admits a star $k$-coloring.
  \item $K(e) \leq k$ for every $e \in E(T)$ and every maximal $(k-1)$-subtree of $(T, K)$ contains a $(k-1)$-exit vertex of $(T, K)$.
\end{enumerate}
\end{theorem}

\begin{proof}
   (i) $\Longrightarrow$ (ii) is trivial.  

   (ii) $\Longrightarrow$ (iii).  Suppose $G(T, K)$ admits a star $k$-coloring $\rho$.  By Observation~\ref{obs:block coloring}, $K(e) \leq k$ for every $e \in E(T)$ and, moreover, $K(e) \leq k-1$ when $e$ incides in a leaf.  Let $R$ be a maximal $(k-1)$-subtree of $(T, K)$, and consider a maximal path $P = v_1, \ldots, v_j$ of $R$ such that $\rho(v_i) \not\in \rho(B(v_iv_{i+1}) \cup \{v_{i+1}\})$ for $1 \leq i < j$.   Note that if $j > 1$, then $|B(v_iv_{i+1})| = k-1$, thus (1) $\rho(v_j) \in \rho(B(v_{j-1}v_j))$.   We claim that $v_j$ is a $(k-1)$-exit vertex.  Indeed, if $v_j$ is a leaf of $T$ and $w$ is its unique neighbor in $T$, then $j = 1$ by (1) and the fact that $B(wv_j) \cup \{v_j\}$ is a block.  Consequently, by the maximality of $P$, it follows that $K(v_jw) \neq k-1$, i.e., $v_j$ is a $(k-1)$-exit vertex.  On the other case, if $v_j$ is an internal vertex of $T$, then, by Observation~\ref{obs:blockbicliques}, $\{v_j\}S$ is a maximal star of $G(T, K)$ for every $S$ that contains exactly one vertex of $K(v_jw)$ for each $w \in N_T(v)$.  Therefore, there exists $w \in N_T(v_j)$ such that (2) $\rho(v_j) \not\in \rho(B(v_jw) \cup \{w\})$.  By (1), $w \neq v_{j-1}$, thus $K(v_jw) \neq k-1$ by the maximality of $R$ and $P$.  Moreover, $K(v_jw) \neq k$ by (2), thus $v_j$ is a $(k-1)$-exit vertex.
   
   (iii) $\Longrightarrow$ (i).  For every maximal $(k-1)$-subtree $R$ of $(T, K)$, let $z(R)$ be a $(k-1)$-exit vertex of $R$.  For every $v \in V(R)$, define $z(v) = z(R)$.  Observe this definition is correct, because maximal $(k-1)$-subtrees are vertex-disjoint.  Let $L$ be a $k$-list assignment of $G(T,K)$, and define $\rho$ as an $L$-coloring of $G(T, K)$ satisfying the following properties for every $vw \in E(T)$.  
  \begin{enumerate}[(1)]
    \item $\rho(v) \neq \rho(w)$,
    \item $|\rho(B(vw))| = K(vw)$ and $\rho(B(vw) \cup \{v,w\}) = \min\{K(vw)+2, k\}$, and
    \item if $K(vw) = k-1$ and $w$ belongs to the unique path from $v$ to $z(v)$ in $T$, then $\rho(v) \not\in \rho(B(vw))$.
  \end{enumerate}
  A coloring satisfying the above conditions can obtained by first coloring the vertices of $T$, and then coloring the vertices in $B(e)$ for every $e \in E(T)$.  Observe, in particular, that if $K(vw) = k-1$, then $z(v) = z(w)$, thus (3) is always possible.  We claim that $\rho$ is a star $L$-coloring of $G(T, K)$.  Let $\{v\}S$ be a maximal star of $G(T, K)$.  Suppose first that $S = \{x\}$, thus either $\{v, x\} \subseteq B(wz)$ for $wz \in V(T)$ or $v$ is a leaf of $T$ and $x \in B(vw) \cup \{w\}$ for some $w \in T$.  In the former case $\rho(v) \neq \rho(x)$ by (2).  In the latter case, by (2), either $\rho(v) \neq \rho(x)$ or $K(vw) \geq k-1$.  If $K(vw) \geq k-1$, then $K(vw) = k-1$ and $w$ is not a leaf of $T$; otherwise one of $\{v\}$ and $\{vw\}$ would induce be a maximal $(k-1)$-subtree without $(k-1)$-exit vertices.  Hence, $\rho(v) \neq \rho(x)$ by (3).  Suppose now that $|S| > 1$, thus $v$ is an internal vertex of $T$ by Observation~\ref{obs:blockbicliques}.  By hypothesis, $v$ has some neighbor $w \in V(T)$ such that either $K(vw) < k-1$ (when $v = z(v)$) or $w$ belongs to the path of $T$ between $v$ and $z(v)$ (when $v \neq z(v)$).  By Observation~\ref{obs:blockbicliques}, $S$ contains a vertex $x$ of $B(vw)$.  If $v = z(v)$, then $K(vw) < k-1$, thus $\rho(v) \neq \rho(x)$  by (1)~and~(2).  On the other case, if $v \neq z(v)$, then $K(vw) = k-1$, thus $\rho(v) \neq \rho(z)$ by (3).
\end{proof}
 
The algorithmic consequences of Theorem~\ref{thm:block graph coloring} are analogous as those observed for threshold graph in Section~\ref{sec:threshold}.  If $(T, K)$ is a netblock representation and $k$ is maximum among the sizes of the blocks of $G = G(T, K)$, then $\chi_S(G) = ch_S(G)$ equals either $k$ or $k+1$.  When $(T, K)$ is given, it takes $O(|V(T)|)$ time to find all the $(k-1)$-subtrees and its $(k-1)$-exit vertices, if they exist.  Thus, deciding if $\chi_S(G) = k$ takes $O(|V(T)|)$ time when $(T, K)$ is given as input.  Furthermore, if $k < \chi_S(G)$, then it takes $O(|V(T)|)$ time to compute a maximal $(k-1)$-subtree of $(T, K)$ with no $(k-1)$-exit vertices.  Such a subtree can be transformed into an induced subgraph of $G$ in $O(n+m)$ time if required.  Also, a $\chi_S(G)$-star-coloring $\rho$ of $G$ can be obtained in linear with rules (1)--(3) of implication (iii) $\Longrightarrow$ (i).  First apply a BFS traversal of $T$ to color the vertices of $T$ with rule (1), and then color the remaining vertices of $G$ following rules (2) and (3).  Finally, observe that $(T, K)$ can be obtained in $O(n+m)$ time from $G$ when $G$ is encoded by adjacency lists.  Thus, all the discussed algorithms take linear time when applied to the adjacency list representation of $G$.  We record all these observation in the theorem below.

\begin{theorem}\label{thm:block graph complexity}
  Let $(T, K)$ be a netblock representation of a graph $G$.  The problems of computing $\chi_S(G)$, $ch_S(G)$, a maximal $(k-1)$-subtree of $(T, K)$ with no $(k-1)$-exit vertices for $k < \chi_S(G)$, and a $\chi_S(G)$-coloring of $G$ can be solved in $O(|V(T)|)$ time when $(T, K)$ is given as input.  In turn, all these problems take $O(n+m)$ time when an adjacency list representation of $G$ is given.
\end{theorem}

\section{Further remarks and open problems}

In this paper we investigated the time complexity of the star and biclique coloring and choosability problems.  In this section we discuss some open problems that follow from our research.

Theorem~\ref{thm:stcol-2} states that the star $k$-coloring problem is \stp-complete even when the input is restricted to \{$C_4$, $K_{k+2}$\}-free graphs.  In Section~\ref{sec:chordal} we discussed how to generalize this theorem to include \{$C_i$, $K_{k+2}$\}-free graphs, for every $i \in O(1)$.  An interesting question is what happens when $i$ grows to infinity, i.e., what happens when chordal graphs are considered.  By Theorem~\ref{thm:splitnpc}, we know that the star $k$-coloring problem is at least \NP-hard on chordal graphs.  Is it \NP-complete or not?  Similarly, by Theorem~\ref{thm:split stchose}, the star $k$-choosability problem on chordal graphs is at least \ptwop-hard; is it \ptwop-complete?

\begin{openproblem}
 Determine the time complexity of the star $k$-coloring ($k$-choosability) problem on chordal graphs and chordal $K_{k+2}$-free graphs.
\end{openproblem}

To prove the \stp-completeness of the star $k$-coloring problem we showed how to transform a formula $\phi(\VEC x, \VEC y)$ into a graph $G$.  Graph $G$ has many connection vertices that are joined together by $k$-keepers and $k$-switchers.  The purpose of the $k$-keepers is to force two vertices to have the same colors, while $k$-switchers are used to force different colors on a pair of vertices.  Both $k$-keepers and $k$-switchers contain blocks of size $k$.  By taking a close examination at Lemmas \ref{lem:keeper properties}~and~\ref{lem:switcher properties}, it can be seen that these blocks play an important role when colors need to be forced.  An interesting question is whether these blocks can be avoided.  

\begin{openproblem}
 Determine the time complexity of the star $k$-coloring ($k$-choosability) problem on graph where every block has size at most $j$, for $j \leq k$.
\end{openproblem}

Keepers and switchers not only have blocks of size $k$; when a $k$-keeper or a $k$-switcher connects two vertices $v$ and $w$, a clique of size $k+1$ containing $v$ and $w$ is generated.  We know that the star $k$-coloring problem is \stp-complete for $K_{k+2}$-free graphs, but what happens when $K_{k+1}$-free graphs are considered?  The answer for the case $k=2$ is given by Theorem~\ref{thm:stcol triangle-free}, i.e., the star $2$-coloring problem is easy on $K_3$-free graphs.  And for larger values of $k$?

\begin{openproblem}
 Determine the time complexity of the star $k$-coloring ($k$-choosability) problem on $K_j$-free graphs, for $k > 2$ and $4 \leq j \leq k+1$.
\end{openproblem}

In Section~\ref{sec:choosability}, the \ptp-completeness of the star $k$-choosability problem on $C_4$-free graphs is proved by induction.  For the case $k=2$, a graph $G_2$ is built from a DNF formula using the same ideas that we used to prove the hardness of the star $2$-coloring problem.  Then, for the case $k>2$, the graph $G_{k-1}$ is transformed into a graph $G_{k}$.  The reason for splitting the proof in two cases is that long $k$-switchers can no longer be included into $G_k$.  So, to avoid the inclusion of induced $C_4$'s, list switchers are used to build $G_2$, while $k$-switchers are used to transform $G_{k-1}$ into $G_{k}$ for $k > 2$.  This way, each generated graph $G_k$ is $C_4$-free.  We did not find a way to extend the holes in $G_k$ as it is done in Section~\ref{sec:chordal} for the star $k$-coloring problem.  Thus, $G_k$ contains $C_5$ as an induced subgraph for every $k > 2$, while $G_2$ contains $C_7$ as an induced subgraph.  Is the problem simpler when such holes are avoided?

\begin{openproblem}
 Determine the time complexity of the star $k$-choosability problem when all the holes in the input graph have length at least $i$, for $i \in O(1)$.
\end{openproblem}

The star coloring problem is easier when any of the graphs on three vertices does not appear as an induced subgraph, as discussed in Section~\ref{sec:small forbiddens}.  Similarly, the biclique coloring problem is simpler when it contains no induced $P_3$, $\overline{P_3}$, or $\overline{K_3}$, or when it is $K_3$-free and biclique-dominated.  Also, by Theorem~\ref{thm:np w4-dart-gem-kii}, the biclique coloring problem on $K_3$-free graphs is ``only'' \NP when it contains no induced $K_{i,i}$ for $i \in O(1)$.  

\begin{openproblem}
 Determine the time complexity of the biclique $k$-coloring ($k$-choosability) problem on $K_3$-free graphs.
\end{openproblem}

The star and biclique coloring problems are also simplified when diamonds are forbidden, as seen in Section~\ref{sec:diamond-free}.  By Theorem~\ref{thm:np w4-dart-gem}, the star $k$-coloring problem is \NP for diamond-free graphs, while the biclique $k$-coloring problem is \NP for \{diamond, $K_{i,i}$\}-free graphs ($i \in O(1)$) by Theorem~\ref{thm:np w4-dart-gem-kii}.  Additionally, if holes and nets are forbidden, then the star $k$-coloring problem can be solved easily by Theorem~\ref{thm:block graph complexity}.  This leaves at least two interesting questions.

\begin{openproblem}
 Determine the time complexity of the biclique $k$-coloring problems on diamond-free graphs.
\end{openproblem}

\begin{openproblem}
 Determine the time complexity of the star $k$-choosability problems on diamond-free graphs.
\end{openproblem}


\end{document}